\documentclass[a4paper,UKenglish,cleveref, autoref, thm-restate]{lipics-v2019}
\nolinenumbers
\usepackage[utf8]{inputenc}
\usepackage{amssymb}
\usepackage{amsmath}
\usepackage{amsthm, amsfonts}
\usepackage{booktabs}

\usepackage{algorithm}
\usepackage{algorithmic}
\usepackage{tikz}
\usepackage{xcolor}

\newtheorem{theorem*}{Theorem}
\newtheorem{lemma*}{Lemma}

\mathchardef\mhyphen="2D
\newcommand{\CS}{\mathrm{CS}}
\newcommand{\CA}{\mathrm{CA}}
\newcommand{\MA}{\mathrm{MA}}
\newcommand{\TA}{\mathrm{TA}}

\title{Approximation algorithms for car-sharing problems}

\author{Kelin Luo\footnote{corresponding author}}{Department of Mathematics and Computer Science, Eindhoven University of Technology, Eindhoven, the Netherlands}{k.luo@tue.nl}{https://orcid.org/0000-0003-2006-0601}{This project has received funding from the European Union's Horizon 2020 research and innovation programme under the Marie Skłodowska-Curie grant agreement number 754462.}
\author{Frits C. R. Spieksma}{Department of Mathematics and Computer Science, Eindhoven University of Technology, Eindhoven, the Netherlands}{k.luo@tue.nl}{https://orcid.org/0000-0002-2547-3782}{This project has received funding from the NWO Gravitation Project NETWORKS, Grant Number 024.002.003.}
\authorrunning{K. Luo and F.C.R. Spieksma}
\Copyright{K. Luo and F.C.R. Spieksma}
\ccsdesc{Theory of computation~Approximation algorithms analysis}
\keywords{car-sharing, approximation Algorithms, matching}
\relatedversion{}

\begin{document}

\maketitle
\begin{abstract}
We consider several variants of a car-sharing problem. Given are a number of requests each consisting of a pick-up location and a drop-off location, a number of cars, and nonnegative, symmetric travel times that satisfy the triangle inequality. Each request needs to  be served by a car, which means that a car must first visit the pick-up location of the request, and then visit the drop-off location of the request. Each car can serve two requests. One problem is to serve all requests with the minimum total travel time (called $\CS_{sum}$), and the other problem is to serve all
requests with the minimum total latency (called $\CS_{lat}$). We also study the special case where the pick-up and drop-off location of a request coincide. We propose two basic algorithms, called the match and assign algorithm and the transportation algorithm. We show that the best of the resulting two solutions is a $ 2$-approximation for $\CS_{sum}$ (and a $7/5$-approximation for its special case), and a $5/3 $-approximation for $\CS_{lat}$ (and a $3/2$-approximation for its special case); these ratios are better than the ratios of the individual algorithms. Finally, we indicate how our algorithms can be applied to more general settings where each car can serve more than two requests, or where cars have distinct speeds. 
\end{abstract}

\section{Introduction}
\label{sec:introduction}

We investigate the following car-sharing problem: there are $n$ cars (or servers), with car $k$ stationed at location $ d_k$ ($ 1\le k\le n$), and
there are $2n$ requests, each request $i$ consisting of a pick-up location $s_i$ and a drop-off location $t_i$ ($1 \le i\le 2n$). Between every pair of locations $x$ and $y$, a number $w(x,y)$ is given; these numbers can be interpreted as the distance between locations $x$ and $y$ or the time needed to travel between $x$ and $y$. These distances, or times, are non-negative, symmetric, and satisfy the triangle inequality. Each request needs to  be served by a car, which means that a car must first visit the pick-up location, and then visit the drop-off location. The car sharing problem is to assign all requests to the cars such that each car serves exactly two requests while minimizing total travel time, and/or while minimizing total waiting time (called total latency), incurred by customers that submitted the requests. Let us elaborate on these two objectives.

\begin{itemize}
\item \textbf{Minimize total travel time.}  The total travel time is the travel time each car drives to serve its requests, summed over the cars. From the company or drivers' perspective, this is important since this objective reflects minimizing costs while serving all requests. From a societal point of view, this objective also helps to reduce emissions.
We use $\CS_{sum}$ to refer to the car-sharing problem with the objective to minimize total travel time. The special case of $\CS_{sum}$ where the pick-up and drop-off location is identical for each request, is denoted by $\CS_{sum, s=t}$.

\item \textbf{Minimize total latency.}  The total latency represents the sum of the travel times needed for each individual customer to arrive at her/his drop-off location, summed over the customers. From the customers' perspective, this is important because it helps customers to reach their destinations as soon as possible. We use $\CS_{lat}$ to refer to the car-sharing problem with the objective to minimize total latency. The special case of $\CS_{lat}$ where the pick-up and drop-off location is identical for each request, is denoted by $\CS_{lat, s=t}$.

\end{itemize}

\paragraph*{Motivation}
Consider a working day morning. A large number of requests, each consisting of a pick-up and a drop-off location has been submitted by the customers. The car-sharing company has to assign these requests to available cars. In many practical situations, each request is allowed to occupy at most two seats in a car (see Uber~\cite{ube}). Thus, in a regular car where at most four seats are available, at most two requests can be combined. Since the company knows the location of the available cars, an instance of our problem arises.

Another application can be found in the area of collective transport for specific groups of people. For instance, the company Transvision~\cite{Trans} 
organizes collective transport by collecting all requests in a particular region of the Netherlands, combines them, and offers them to regular transport companies. In their setting customers must make their request the evening before the day of the actual transport; the number of requests for a day often exceeds 5.000. However, in this application, it is true that a car may pick up more than two requests (we come back to this issue in Section~\ref{sec:conclusions}).

The problems $\CS_{sum,s=t}$ and $\CS_{lat,s=t}$ are natural special cases of $\CS_{sum}$ and $\CS_{lat}$ respectively, and can be used to model situations where items have to be delivered to clients (whose location is known and fixed). For instance, one can imagine a retailer sending out trucks to satisfy clients' demand where each truck is used to satisfy two clients. 

\paragraph*{Related Work.}
In the literature, car-sharing systems are increasingly studied. Agatz et al.~\cite{agatz2011time} consider the problem of matching cars and customers in real-time with the aim to minimize the total number of vehicle-miles driven.  
Stiglic et al.~\cite{stiglic2016making} show that a small increase in the flexibility of either the cars or the customers can significantly increase the performance of a dynamic car-sharing system.
Furthermore, Wang et al.~\cite{wang2018stable} introduce the notion of stability for a car sharing system and they present methods to establish stable or nearly stable solutions.
More recently, Ashlagi et al.~\cite{ashlagi2019edge} study the problem of matching requests to cars that arrive online and need to be assigned after at most a prespecified number of time periods. Every pair of requests that is matched yields a profit and the goal of the  car sharing system is to match requests with maximum total profit.

Bei and Zhang~\cite{bei2018algorithms} introduce $\CS_{sum}$, and give a $2.5$-approximation algorithm for it. In fact, both $\CS_{sum}$ and $\CS_{lat}$ are a special case of the so-called two-to-one assignment problem (2-1-AP) investigated by Goossens et al.~\cite{goossens2012between}. Given a set $G$ of $n$ green elements, and a set $R$ of $2n$ red elements, we call a {\em triple} a set of three elements that consist of a single green element and two red elements; further, there is a cost-coefficient for each triple. The 2-1-AP problem is to find a collection of triples covering each element exactly once, while the sum of the corresponding cost-coefficients is minimized.  
In the context of our car sharing problem, the green elements represent the cars, and the red elements represent the requests. For the special case of 2-1-AP where the cost of each triple $ (i,j,k)$ is defined as the sum of the three corresponding distances, i.e., $cost(i,j,k)=d_{ij}+d_{jk}+d_{ki}$, that satisfy the triangle inequality, Goossens et al.~\cite{goossens2012between} give an algorithm with approximation ratio $4/3$. The definition of the cost-coefficients in $\CS_{sum}$, as well as in $\CS_{lat}$ differs from this expression; we refer to Section~\ref{sec:preliminaries} for a precise definition.

\paragraph*{Our results}

We formulate and analyze three polynomial-time approximation algorithms for our car-sharing problems: a match and assign algorithm MA, a transportation algorithm TA, and a combined algorithm CA which runs MA and TA, and then outputs the best of the two solutions. These algorithms are extended versions of algorithms in \cite{goossens2012between} using as additional input a parameter $\alpha \in \{1,2\}$ and a vector $v\in\{u,\mu\}$; for their precise description, we refer to Section~\ref{sec:algorithms}.  All of them run in time $O(n^3)$.
 Here, we establish their worst-case ratio's (see Williamson and Shmoys~\cite{williamson2011design} for appropriate terminology).

\begin{table}[thbp]
	\centering
	\caption{Overview of our results}
	\label{tab1}
	\scalebox{0.85}{
	\begin{tabular}{p{45pt}p{60pt}p{40pt}p{60pt}p{40pt}p{75pt}p{75pt}}
		\toprule
		Problem & MA($1, u$) & MA($2, \mu$) & TA(1) & TA(2) & CA($1, u$) & CA($2, \mu$) \\			
		\midrule
       $CS_{sum}$ & 2$^*$ (\cref{lem:alg_ana_401}) & 3 & 3$^*$ (\cref{lem:alg_ana_402})  & 4  & 2$^*$ (\cref{the:alg_ana_1})  & 3  \\		
	$CS_{sum, s=t}$ & 3/2$^*$  & 3/2  & 3$^*$   & 4  & 7/5$^*$ (\cref{the:alg_ana_3})  & 10/7   \\		
		 $CS_{lat}$ & 4  &  2$^*$  & 3   & 2$^*$  & 3   & 5/3 (\cref{the:alg_ana_2}) \\
		 $CS_{lat, s=t}$ & 2  & 2$^*$    & 2   & 2$^*$  & 8/5    & 3/2$^*$ (\cref{the:Pl1_03}) \\
		\bottomrule			
	\end{tabular}}
\end{table}

For all four problems, we show how the above mentioned algorithms behave with respect to their worst-case ratios.
Notice that for $\CS_{sum, s=t}$, $\CS_{lat}$ and $\CS_{lat, s=t}$,  the worst-case ratio of the combined algorithm CA is strictly better than each of the two worst-case ratios of the algorithms CA consists of. 
An overview of all our results is shown in Table~\ref{tab1}; $u$ and $ \mu$ are defined in Section~\ref{sec:preliminaries}; an ``$^*$'' means that the corresponding worst-case ratio is tight; 
results not attributed to a specific Lemma or Theorem are proven in~\cref{app:tableresults}.

The two problems $\CS_{sum}$ and $\CS_{lat}$ only differ in their objectives, i.e., anything that is a feasible solution in $\CS_{sum}$ is also a feasible solution in $\CS_{lat}$ and vice versa. A similar statement can be made for $\CS_{sum,s=t}$ and $\CS_{lat,s=t}$. Hence, we can use a two-criteria approach to judge a solution; more concretely, we call an algorithm a ($\rho, \varphi$)-approximation algorithm if the cost of its solution is at most $\rho$ times the cost of an optimal solution to $\CS_{sum}$, and at most $\varphi$ times the cost of an optimal solution of $\CS_{lat}$.
Therefore, CA$(1, u)$ is a (2, 3)-approximation (resp. ($\frac{7}{5}$, $\frac{8}{5}$)-approximation) algorithm and CA$(2,  \mu)$ is a (4, $\frac{5}{3}$)-approximation (resp. ($\frac{10}{7}$, $\frac{3}{2}$)-approximation)  algorithm. Notice that, from a worst-case perspective, none of these two algorithms dominates the other.

\section{Preliminaries}\label{sec:prelim}
\label{sec:preliminaries}

\textbf{Notation.} \
 We consider a setting with $n$ cars, denoted by $D=\{1,2,...,n\}$, with car $k$ at location $ d_k$ ($k \in D$). There are $2n$ requests $R=\{1,2,...,2n\}$, where request $i$ is specified by the pick-up location $s_i$ and the drop-off location $t_i$ ($i \in R$).  
The distance (or cost, or travel time) between location $x_1$ and $ x_2$ is denoted by $w(x_1,x_2)$. We assume here that distances $ w(x_1,x_2)$ are non-negative, symmetric, and satisfy the triangle inequality. Furthermore, we extend the notation of distance between two locations to the distance of a path:  $w(x_1,x_2,..., x_k)=\sum^{k-1}_{i=1} w(x_i, x_{i+1})$.
 We want to find an \emph{allocation} $ M=\{(k, R_k):k\in D, R_k\subseteq R\}$, where $|R_k|=2$ and $ R_1, R_2, .., R_n$ are pairwise disjoint; the set $R_k$ contains the two requests that are assigned to car $k$. For each car $ k\in D$ and $ (k, R_k)\in M$ where $ R_k$ contains request $i$ and request $j$, i.e., $R_k=\{i, j\}$, 
 we denote the \emph{travel time} of serving the requests in $R_k$ by:
 \begin{equation}
  \begin{split}
   \label{equation101}
cost(k,\{i,j\})\equiv& \min\{w(d_k,s_i,s_j, t_i,t_j),w(d_k,s_i,s_j, t_j,t_i),w(d_k,s_i,t_i, s_j,t_j),\\
&w(d_k,s_j,s_i, t_i,t_j),w(d_k,s_j,s_i, t_j,t_i),w(d_k,s_j,t_j, s_i,t_i)\}.
  \end{split}
\end{equation}
Notice that the six terms in~(\ref{equation101}) correspond to all distinct ways of visiting the locations $ s_i$, $ s_j$, $ t_i$ and $ t_j$ where $ s_i$ is visited before $t_i$ and $ s_j$ is visited before $t_j$.
We view $cost(k, \{i,j\})$ as consisting of two parts: one term expressing the travel time between $d_k $ and the first pick-up location $ s_i$ or $ s_j$; and another term capturing the travel time from the first pick-up location to the last drop-off location. For convenience, we give a definition of the latter quantity.
 \begin{equation}
  \begin{split}
   \label{equation101_u}
 u_{ij} \equiv \min\{w(s_i,s_j, t_i,t_j),w(s_i,s_j, t_j,t_i),&w(s_i,t_i, s_j,t_j) \} \\
&\mbox{ for each } i,j \in R \times R, i \neq j.
  \end{split}
\end{equation}
Notice that the $ u_{ij}$'s are not necessarily symmetric. For $CS_{sum, s=t}$, $u_{ij}\equiv w(s_i, s_j)$.
Using (\ref{equation101}) and (\ref{equation101_u}), the travel time needed to serve requests in $R_k = \{i,j\}$ ($k\in D$) is then given by:
\begin{equation}
   \label{equation101_u_1}
cost(k,\{i,j\}) = \min\{w(d_k,s_i)+u_{ij},w(d_k,s_j)+u_{ji}\}.
\end{equation}
We denote the travel time of an allocation $M$ by:
\begin{equation}
\label{eq:totalcost}
cost(M)=\sum_{(k, R_k)\in M} cost (k, R_k).
\end{equation}
In $\CS_{sum}$ and $\CS_{sum, s=t}$, the goal is to find an allocation $M$ that minimizes $cost(M)$.

Let us now consider $\CS_{lat}$. Here, we focus on the waiting time as perceived by an individual customer, from the moment the car leaves its location until the moment the customer reaches its drop-off location. More formally, we denote the \emph{latency} of serving requests in $R_k = \{i,j\}$ by:
 \begin{equation}
   \begin{split}
   \label{equation102}
wait&(k,\{i,j\})  \equiv \min \\
&\{w(d_k,s_i, s_j,t_i)+ w(d_k, s_i, s_j,t_i, t_j), w(d_k,s_i, s_j,t_j)+ w(d_k,s_i, s_j,t_j, t_i), \\
& w(d_k,s_i, t_i)+w(d_k,s_i, t_i, s_j,t_j), w(d_k,s_j, s_i,t_i)+ w(d_k,s_j, s_i,t_i, t_j),\\
& w(d_k,s_j, s_i,t_j)+ w(d_k,s_j, s_i,t_j, t_i), w(d_k,s_j, t_j)+w(d_k,s_j, t_j, s_i,t_i)\}.
  \end{split}
\end{equation}

Again, the six terms in~(\ref{equation102}) correspond to all distinct ways of visiting the locations $s_i$, $ s_j$, $ t_i$ and $ t_j$ where $ s_i$ is visited before $t_i$ and $ s_j$ is visited before $t_j$. We view $wait(k, R_k)$ as consisting of two parts: one term expressing the waiting time between $d_k $ and the first pick-up location $ s_i$ or $ s_j$; and another term capturing the waiting time from the first pick-up location to the last drop-off location. For convenience, we give a definition of the latter quantity.
 \begin{equation}
  \begin{split}
   \label{equation102_u}
 \mu_{ij}\equiv \min\{&w(s_i,s_j, t_i)+w(s_i,s_j, t_i,t_j),w(s_i,s_j, t_j)+w(s_i,s_j, t_j,t_i),\\
&w(s_i, t_i)+w(s_i,t_i, s_j,t_j) \} \mbox{ for each } i,j \in R \times R, i \neq j.
  \end{split}
\end{equation} 
Notice that the $ \mu_{ij}$'s are not necessarily symmetric. For $CS_{lat, s=t}$, $\mu_{ij}\equiv w(s_i, s_j)$.
Using (\ref{equation102}) and (\ref{equation102_u}), the latency needed to serve requests in $R_k = \{i,j\}$ ($k\in D$) is then given by:
\begin{equation}
   \label{equation102_u_1}
wait(k,\{i,j\})= \min\{2w(d_k,s_i)+\mu_{ij},2w(d_k,s_j)+\mu_{ji}\}.
\end{equation}
We denote the latency of an allocation $M$ by: \begin{equation}
\label{eq:totallatency}
wait(M)=\sum_{(k, R_k)\in M} wait(k, R_k).
\end{equation}
Thus, in $\CS_{lat}$ and $\CS_{lat, s=t}$, the goal is to find an allocation $M$ that minimizes $wait(M)$.

A natural variant of the latency objective is one where the latency is counted with respect to the pick-up location as opposed to the drop-off location in $\CS_{lat}$. Clearly, then the drop-off location becomes irrelevant, and in fact our approximation results for $\CS_{lat, s=t}$ become valid for this variant.

\textbf{Remark.} \ Bei and Zhang~\cite{bei2018algorithms} prove that the special car-sharing problem, i.e., $\CS_{sum, s=t}$ is NP-hard. Their  proof can also be used to prove that $\CS_{lat, s=t}$ is NP-hard. In fact, we point out that the arguments presented in Goossens et al.~\cite{goossens2012between} allow to establish the APX-hardness of these two problems. 

\textbf{Paper Outline} \
In~\cref{sec:algorithms}, we present two algorithms, i.e., the match and assignment algorithm and the transportation algorithm. In~\cref{sec:approximation_results}, we prove the approximation ratios.~\cref{sec:conclusions} concludes the paper. 

\section{Algorithms}
\label{sec:algorithms}
We give three polynomial-time approximation algorithms for our car-sharing problems. In Section~\ref{subsec:MAalgorithms} we describe the match and assign algorithm MA($\alpha, v$), and in Section~\ref{subsec:TAalgorithms} we describe the transportation algorithm TA($ \alpha$). As described before, the third algorithm, CA($\alpha, v$), simply runs MA($\alpha, v$) and TA($ \alpha$), and then outputs the best of the two solutions. Let $\alpha \in\{1,2\} $ be the coefficient that weighs the travel time between the car location and the first pick-up location, and let $v \in \{u,\mu\}$ (where $u$ is defined in~(\ref{equation101_u}), and $\mu$ is defined in~(\ref{equation102_u})).

\subsection{The match and assign algorithm}
\label{subsec:MAalgorithms}
 
The match-and-assign algorithm MA$(\alpha, v)$ goes through two steps: in the first step, the algorithm pairs the requests based on their combined serving cost, and in the second step, the algorithm assigns the request-pairs to the cars.
 
 \begin{algorithm}	
	\caption{Match-and-assign algorithm (MA$(\alpha, v)$)}
	\label{alg:MA}
	\begin{algorithmic}[1]	
	\STATE \emph{\normalsize Input}: non-negative weighted graph $G=(V,E,w)$, requests $R=\{i=(s_i,t_i): 1\le i\le m, s_i, t_i\in V\}$, cars $ D=\{k:1\le k\le n, d_k\in V\}$,
	$ \alpha \in\{1,2\} $ and $ v \in\{u, \mu\} $.\\
\STATE \emph{\normalsize Output}: An allocation MA$=\{(k,\{i,j\}):k\in D, i,j\in R\} $.\\
\STATE For $i, j \in R$ do\\
\STATE \ \ \ \ \ \ \ \  $v_1(\{i,j\}) \equiv \frac{v_{ij}+v_{ji}}{2}$\\
\STATE  end for\\
\STATE  Let $G_1\equiv (R,v_1)$ be the complete weighted graph where an edge between vertex $ i\in R$ and vertex $j\in R $ has weight $v_1(i,j)$. \\
\STATE Find a minimum weight perfect matching $M_1$ in $G_1\equiv (R,v_1)$  with weight $ v_1(M_1)$.\\
\STATE For $k \in D$ and $ \{i,j\} \in M_1$  with $ v_{ij}\ge v_{ji}$ do\\
\STATE \ \ \ \ \ \ \ \  $v_2(k,\{i,j\}) \equiv \min \{\alpha w(d_k,s_i)+\frac{v_{ij}-v_{ji}}{2}, \alpha w(d_k,s_j)-\frac{v_{ij}-v_{ji}}{2}\}$\\
\STATE end for\\
\STATE Let $G_2\equiv (D\cup M_1,v_2)$ be the complete bipartite graph with \emph{left} vertex-set $D$, \emph{right} vertex-set $M_1$ and edges with weight $v_2(k,\{i,j\})$ for $k\in D$, and $\{i,j\}\in M_1 $.\\
\STATE Find a minimum  weight perfect matching $M_2$ in $G_2\equiv (D\cup M_1,v_2)$  with weight $ v_2(M_2)$. \\
\STATE Output allocation MA$=M_2$.
	\end{algorithmic}		
\end{algorithm}
 
The key characteristics of algorithm MA$(\alpha, v)$ are found in lines 4 and 9 where the costs of the first and second step are defined. A resulting quantity is $v_1(M_1) + v_2(M_2)$; we now prove two lemma's concerning this quantity, which will be of use in Section~\ref{sec:approximation_results}.

\begin{lemma}
\label{lemma:ma} For each $\alpha \in \{1,2\}$ and $v \in \{u, \mu\}$, we have:
\[ v_1(M_1)+v_2(M_2) = \sum_{(k,\{i,j\})\in M_2} \min \{ \alpha w(d_k,s_i)+ v_{ij}, \alpha w(d_k,s_j)+ v_{ji}\}.\]
\end{lemma}
\begin{proof}
Without loss of generality, for any $ \{i,j\}\in M_1$, suppose $v_{ij}-v_{ji}\ge 0$ (the other case is symmetric). 
\begin{equation*}
   \begin{split}
   \label{equation_lemma:ma}
& v_1(M_1)+v_2(M_2) \\
&=  \sum_{\{i,j\}\in M_1} \frac{v_{ij}+v_{ji}}{2}  +\sum_{(k,\{i,j\})\in M_2} \min \{\alpha w(d_k,s_i)+\frac{v_{ij}-v_{ji}}{2}, \alpha w(d_k,s_j)-\frac{v_{ij}-v_{ji}}{2}\}  \\
&=\sum_{(k,\{i,j\})\in M_2} \min \{\alpha w(d_k,s_i)+v_{ij}, \alpha w(d_k,s_j)+v_{ji} \}.
  \end{split}
\end{equation*}
The first equality follows from  lines 4 and 9 in Algorithm~\ref{alg:MA}.
\end{proof}

\begin{lemma}
\label{lemma:ma_min}
For $\alpha \in \{1,2\}$, $v \in \{u, \mu\}$, and for each allocation $M$, 
we have:
\[ v_1(M_1)+v_2(M_2)\le \frac{1}{2} \sum_{(k,\{i,j\})\in M} (\alpha (w(d_k,s_i)+ w(d_k,s_j))+ v_{ij}+ v_{ji}).\]
\end{lemma}
\begin{proof}
 For an allocation $M$, let $ M_R=\{R_k: (k, R_k)\in M\}$. Observe that
 \begin{equation}
   \label{equation_cycleequation1}
   v_1(M_1) \leq \sum_{\{i,j\}\in M_R} \frac{v_{ij}+v_{ji}}{2},
\end{equation}
since $M_1$  is a minimum weight perfect matching in $G_1\equiv (R,v_1)$.

We claim that \begin{equation}
   \label{equation_cycleequation2}
   v_2(M_2)\leq  \frac{1}{2} \sum_{(k,\{i,j\})\in M} \alpha (w(d_k,s_i)+w(d_k,s_j)).
\end{equation}
When summing~(\ref{equation_cycleequation1}) and~(\ref{equation_cycleequation2}), the lemma follows.

 Hence, it remains to prove~(\ref{equation_cycleequation2}).
Consider an allocation $M$, and consider the matching $M_1$ found in the first step of MA.
Based on $M$ and $M_1$, we construct the graph $ G'=(R\cup D, M_1\cup\{(\{i,k\},\{j,k\}):( k,\{i,j\})\in M\})$. Note that every vertex in graph $G'$ has degree 2. Thus, we can partition $ G'$ into a set of disjoint cycles called $C$; each cycle $c \in C$ can be written as $c=(i_1, j_1, k_1, i_2, j_2, k_2, ...., k_h, i_1)$, where $\{i_s, j_s\}\in M_1$, $(k_s, \{j_s, i_{s+1}\})\in M $ for $ 1\le s<h$ and $(k_h, \{j_h, i_{1}\})\in M$. Consider now, for each cycle $c \in C$, the following two matchings called $M^c_{\ell}$ and $M^c_r$:
 \begin{itemize}
  \item $M^{c}_{\ell} =\{(\{i_1,j_1\}, k_1), (\{i_2,j_2\}, k_2), ...,(\{i_h,j_h\}, k_h)\} $,
  \item $M^{c}_r =\{(k_1, \{i_2,j_2\}), (k_2, \{i_3,j_3\}), ...,(k_h, \{i_1,j_1\})\} $.
\end{itemize}
Obviously, both $M_{\ell}\equiv \bigcup_{c \in C} M^{c}_{\ell} $, and $M_r\equiv \bigcup_{c \in C} M^{c}_r  $ are a perfect matching in $G_2=(D\cup M_1, v_2)$.
Given the definition of $v_2(k, \{i,j\})$ (see line 9 of Algorithm MA), we derive for each pair of requests $ \{i,j\}$ and two cars $a, b$:
$v_2(a, \{i,j\})+ v_2(b, \{i,j\})\le \alpha w(d_{a},s_i)+\frac{v_{ij}-v_{ji}}{2}+\alpha w(d_{b},s_j)-\frac{v_{ij}-v_{ji}}{2}= \alpha (w(d_{a},s_i)+ w(d_{b},s_j)).$  
Similarly, it follows that: $ v_2(a, \{i,j\})+ v_2(b, \{i,j\})\le  \alpha (w(d_{a},s_j)+ w(d_{b},s_i))$.
Thus, for each $c \in C$:
\begin{equation}
\label{ineq:lemma2}
\hspace*{-0.5cm}
\sum\limits_{(k, \{i,j\})\in M^{c}_{\ell}} v_2(k, \{i,j\})  + \sum\limits_{(k, \{i,j\})\in M^{c}_r} v_2(k, \{i,j\}) \le \sum\limits_{\substack{\{i,k\}, \{j,k\} \in c \\ (k,\{i,j\})\in M}} \alpha (w(d_k,s_i)+ w(d_k,s_j)).\end{equation}

Note that $M_2$ is a minimum weight perfect matching in $ G_2=(D\cup M_1, v_2)$, and both $M_{\ell}$ and $M_r$ are a perfect matching in $ G_2=(D\cup M_1, v_2)$. Thus:
\begin{equation*}
   \begin{split}
   \label{equation010}
 v_2(M_2)&\le \sum_{c \in C} \min \{ v_2(M^{c}_{\ell}), v_2(M^{c}_r)\}\\
&\le \frac{1}{2}  \sum_{c \in C} (\sum_{(k, \{i,j\})\in M^{c}_{\ell}} v_2(k, \{i,j\})  +  \sum_{(k, \{i,j\})\in M^{c}_r} v_2(k, \{i,j\})) \\
&\le \frac{1}{2} \sum_{(k, \{i,j\})\in M} \alpha( w(d_k,s_i)+ w(d_k,s_j)). \\
  \end{split}
\end{equation*}
The last inequality follows from (\ref{ineq:lemma2}), and hence (\ref{equation_cycleequation2}) is proven.
\end{proof}

\subsection{The transportation algorithm}
\label{subsec:TAalgorithms}
In this section, we present the transportation algorithm.  
The idea of the algorithm is to replace each car $k\in D$ by $2$ virtual cars called $\gamma(k)$ and $\delta(k)$, resulting in two car sets $\Gamma =\{\gamma(1), ..., \gamma(n)\}$ and $\Delta=\{\delta(1), ..., \delta(n)\}$. Next we assign the requests to the $2n$ cars using a particular definition of the costs; a solution is found by letting car $k \in D$ serve the requests assigned to car $\gamma(k)$ and $\delta(k)$.  

 \begin{algorithm}	
	\caption{Transportation algorithm (TA$(\alpha)$)}
	\label{alg:TA}
	\begin{algorithmic}[1]
	\STATE \emph{\normalsize Input}: non-negative weighted graph $G=(V,E,w)$, requests $R=\{i=(s_i,t_i): 1\le i\le m, s_i, t_i\in V\}$, cars $ D=\{k: 1\le k \le n, d_k\in V\}$, two virtual car sets $\Gamma=\{\gamma(1), ..., \gamma(n)\}, $ and  $\Delta=\{\delta(1), ..., \delta(n)\}$, and $ \alpha \in\{1,2\} $.\\
\STATE \emph{\normalsize Output}: An allocation TA$=\{(k,\{i,j\}):k\in D, i,j\in R\} $.\\
\STATE For $k\in D, i \in R$ do\\
\STATE \ \ \ \ \ \ \ \  $v_3(\gamma(k),i)=\alpha w(d_k, s_i, t_i)+ w(t_i,d_k)$\\
\STATE \ \ \ \ \ \ \ \  $v_3(\delta(k),i)=w(d_k, s_i, t_i)$\\
\STATE end for\\
\STATE Let $G_3\equiv (\Gamma \cup \Delta \cup R, v_3)$ be the complete bipartite graph with \emph{left} vertex-set $\Gamma \cup \Delta $, \emph{right} vertex-set $R$, and edges with weight $v_3(x,i)$ for $ x\in \Gamma \cup \Delta $ and $i\in R $. \\
\STATE Find a minimum weight perfect matching $M_3$ in  $G_3\equiv (\Gamma \cup \Delta \cup R, v_3)$ with weight $ v_3(M_3)$.\\
\STATE  Output allocation TA$=M_4\equiv\{(k,\{i,j\}): (\gamma(k),i), (\delta(k), j)\in M_3, k \in D \}$.
	\end{algorithmic}		
\end{algorithm}

The crucial points of algorithm TA$(\alpha)$ are found in lines 4 and 5 where the costs of assigning a request are defined; a resulting quantity is $v_3(M_3)$.  
We now prove two lemma's concerning this quantity $v_3(M_3)$, which will be of use in Section~\ref{sec:approximation_results}.

\begin{lemma}
\label{lemma:ta}
For each $\alpha \in \{1,2\}$, we have:
\[v_3(M_3)= \sum_{(k,\{i,j\})\in M_4} \min \{\alpha w(d_k, s_i, t_i)+w(t_i, d_k, s_j, t_j), \alpha w(d_k, s_j, t_j)+w(t_j, d_k, s_i, t_i)\}.\]
\end{lemma}

This lemma follows directly from the definition of the costs in lines 4 and 5 of TA.

\begin{lemma}
\label{lemma:ta_min}
For $\alpha \in \{1,2\}$ and for each allocation $ M$, we have:
 \[ v_3(M_3)\le \sum_{(k,\{i,j\})\in M} \min \{\alpha w(d_k, s_i, t_i)+ w(t_i,d_k, s_j, t_j), \alpha w(d_k, s_j, t_j)+ w(t_j,d_k, s_i, t_i)\}.\]
\end{lemma}

This lemma is obvious since $M_3$ is a minimum weight perfect matching.

\textbf{Remark.} \ Both MA($\alpha, v$) and TA($\alpha$) runs in time $O(n^3)$ since a minimum matching $M$ in a weighted graph of $n$ vertices can be found in time $O(n^3)$~\cite{gabow1990data}.

\section{Approximation results}
\label{sec:approximation_results}
In this section, we analyze the combined algorithm CA($\alpha, v$), i.e., the best of the two algorithms MA($\alpha, v$) and TA($\alpha$), for $\CS_{sum}$, $\CS_{sum,s=t}$, $\CS_{lat}$ and $\CS_{lat,s=t}$. 
We denote the allocation by the match and assign algorithm MA($\alpha, v$) (resp. the transportation algorithm TA($\alpha$)) by MA (resp. TA).
We denote an optimal allocation of a specific problem by $ M^*=\{(k, R^*_k): k\in D\}$ with $ R^*_k=\{i,j\}$ ($i,j \in R$). 
Let $ M^*_R=\{R^*_k: (k, R^*_k)\in M^*\}$ denote the pairs of requests in $ M^*$.  
With a slight abuse of notation, we use CA($I$) to denote the allocation found by CA for instance $I$.
 
\subsection{Approximation results for $\CS_{sum}$}
\label{subsec:approximation_results_CS_sum} 
We first establish the worst-case ratios of MA($1, u$) and TA($1$), and next prove that CA($1, u$) is a $ 2$-approximation algorithm.   
\begin{lemma}
\label{lem:alg_ana_401}
$\MA(1, u)$ is a $2$-approximation algorithm for $\CS_{sum}$.
\end{lemma}
\begin{proof}
We assume wlog that, for each $(k,\{ i,j\})\in M^*$, $cost(k,\{i,j\})=w(d_k, s_i)+u_{ij} $. We have:

 \begin{equation}
   \begin{split}
   \label{equation401}
&cost(\text{MA($1, u$)})=\sum_{(k,\{i,j\})\in \text{MA}} \min\{w(d_k, s_i) +u_{ij}, w(d_k, s_j)+u_{ji}\}  \ \   \text{(by~(\ref{equation101_u_1}) and (\ref{eq:totalcost}))} \\
&= v_1(M_1)+ v_2(M_2) \ \ \ \text{(by~\cref{lemma:ma})} \\
&\leq  \frac{1}{2} \sum_{(k,\{ i,j\})\in M^*} (w(d_k,s_i)+w(d_k,s_j)+u_{ij}+u_{ji})  \ \    \text{(by~\cref{lemma:ma_min})} \\
&\leq  \frac{1}{2} \sum_{(k,\{ i,j\})\in M^*} (2w(d_k,s_i)+w(s_i, s_j)+3u_{ij}) \ \    
\\
&\leq  \frac{1}{2} \sum_{(k,\{ i,j\})\in M^*} (2w(d_k,s_i)+4u_{ij}) \ \    \text{(since $w(s_i,s_j)\le u_{ij}$)} \\
&\leq  \frac{1}{2} \sum_{(k,\{ i,j\})\in M^*} 4 cost(k,\{i,j\}) \ \    \text{(by the assumption $cost(k,\{i,j\})=w(d_k, s_i)+u_{ij}$)} \\
&=   2 \ cost(M^*).
  \end{split}
\end{equation}
The second inequality follows from the triangle inequality, and since $u_{ji}\le 2u_{ij}$ for each request pair $\{i,j\}\in R^2$; the corresponding proof can be found in~\cref{A0_1}. \end{proof}

Notice that the statement in Lemma~\ref{lem:alg_ana_401} is actually tight by the instance depicted in Figure~\ref{fig401}.

\begin{lemma}
\label{lem:alg_ana_402}
$\TA(1)$ is a $3$-approximation algorithm for $\CS_{sum}$.
\end{lemma}

\begin{proof}
We assume wlog that, for each $(k,\{ i,j\})\in M^*$, $cost(k,\{i,j\})=w(d_k, s_i)+u_{ij} $. We have:
\begin{equation}
  \begin{split}
   \label{equation402}
cost(\text{TA($1$)})&=\sum_{(k,\{i,j\}) \in \text{TA}} \min\{w(d_k, s_i)+u_{ij}, w(d_k, s_j)+u_{ji}\}  \ \    \text{(by~(\ref{equation101_u_1}))} \\
&\leq  v_3(M_3)  \  \ \text{(by~\cref{lemma:ta})}\\
 &\leq \sum_{(k,\{i,j\})\in M^*}  w(d_k,s_i, t_i,d_k,s_j, t_j) \ \    \text{(by~\cref{lemma:ta_min})} \\
  &\leq \sum_{(k,\{i,j\})\in M^*} 3 cost(k, \{i,j\}) \ \    \text{(by $cost(k,\{i,j\})=w(d_k, s_i)+u_{ij}$)} \\
   &= 3 \ cost(M^*)
     \end{split}
\end{equation}
The third inequality holds since $3 cost(k, \{i,j\})\ge w(d_k,s_i, t_i,d_k,s_j, t_j)$ for each $(k, \{i,j\}) \in R$. The proof can be found in~\cref{A1}. \end{proof}

Notice that the statement in Lemma~\ref{lem:alg_ana_402} is actually tight by the instance depicted in Figure~\ref{fig404} restricted to the subgraph induced by the nodes $(s_5, s_6)$, $(k_3, k_4)$, and $(s_7,s_8)$.

Recall that the combined algorithm CA($1, u$) runs MA($1, u$) and TA($1$), and then outputs the best of the two solutions. We now state the main result of this section.

\begin{theorem}
\label{the:alg_ana_1}
$\CA(1, u)$ is a $2$-approximation algorithm for $\CS_{sum}$. Moreover, there exists an instance $I$ for which  cost$(\CA(I))$ = 2 cost$(M^*(I))$.
\end{theorem}
\begin{proof}
It is obvious that, as cost(CA($1,u)) = \mbox{min}\{$cost(MA($1,u)),$cost(TA(1))\}, Lemma's \ref{lem:alg_ana_401} and \ref{lem:alg_ana_402} imply that CA$(1,u)$ is a 2-approximation algorithm for $\CS_{sum}$. We now provide an instance for which this ratio is achieved.

Consider the instance $I$ depicted in~\cref{fig401}. This instance has $n=2$ with $D=\{1,2\}$ and $R=\{1,2,3,4\}$. Locations corresponding to distinct vertices in \cref{fig401} are at distance $1$.
Observe that an optimal solution is $M^*(I) =  \{(k_1, \{1, 3\}), (k_2, \{2, 4\})\}$ with cost($M^*(I))=2$. Note that $ M^*_R=\{\{1,3\}, \{2,4\}\}$. Let us now analyse the performance of MA($1, u$) and TA($1$) on instance $I$.

\begin{figure}
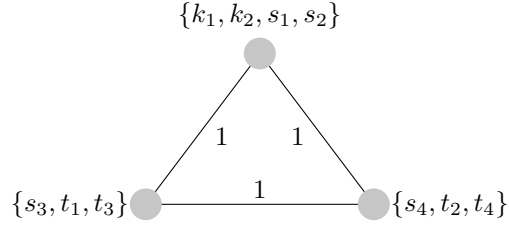

\centering
\tikz \draw (2,3) node[circle]{$\{k_1, k_2, s_1, s_2\}$} (-0.5,0.5) node[circle]{$\{s_3,t_1, t_3\}$} (4.5,0.5) node[circle]{$\{s_4, t_2, t_4\}$} (1.5,1.4) node[circle]{1} (2.5,1.4) node[circle]{1} (2,0.7) node[circle]{1} (0.5,0.5) node[circle, fill=gray!45]{\ \ } -- (2,2.5) node[circle, fill=gray!45]{\ \ }--(3.5, 0.5) node[circle, fill=gray!45]{\ \ }-- cycle;
\caption{A worst-case instance for the combined algorithm CA($ 1, u$) of $\CS_{sum}$.}
\label{fig401}
\end{figure}

Based on the $u_{ij}$ values as defined in (\ref{equation101_u_1}), MA($1, u$) can find, in the first step, matching $M_1=\{\{1,2\}, \{3,4\}\}$ with $v_1(M_1)=3$. Then, no matter how the second step matches the pairs to cars (since two cars stay at the same location), the total cost of MA($1, u$) will be $4$.
TA($1$) can assign request $1$ to car $\gamma(1)$, and request $2$ to car $\delta(1)$, and similarly, request $3$ to car $\gamma(2)$, and request $4$ to car $\delta(2)$. Note that $v_3(\{(k_1,1), (k_1,2), (k_2,3),(k_2,4)\})=v_3(\{(k_1,1), (k_1,3), (k_2,2),(k_2, 4)\})=6$. Thus the total cost of TA($1$) is  $4$.

To summarize, the instance in~\cref{fig401} is a worst-case instance for the combined algorithm CA$(1,u)$.  
\qed \end{proof}

\subsection{Approximation results for $\CS_{sum,s=t}$}
\label{subsec:approximation_results_CS_sumst}
Let us now consider the special case of $\CS_{sum}$ where the pick-up and drop-off location is identical to each request, $\CS_{sum,s=t}$.
From~(\ref{equation401}) and (\ref{equation402}), we have cost(MA($1, u$))$= v_1(M_1)+ v_2(M_2)$ and cost(TA($1$))$\le  v_3(M_3)$. Note that  $ u_{ij}=u_{ji}=w(s_i,s_j)$ in $\CS_{sum,s=t}$.

Now, we show that for $\CS_{sum,s=t}$ algorithm CA$(1,u)$ is a $\frac{7}{5}$-approximation algorithm, a ratio which is strictly better than the ratio's of MA$(1,u)$ and TA(1) for this problem.

 \begin{theorem}
\label{the:alg_ana_3}
$\CA(1, u)$ is a 7/5-approximation algorithm for $CS_{sum,s=t}$. Moreover, there exists an instance $I$ for which  cost$(\CA(I))$=7/5 cost$(M^*(I))$.
\end{theorem}
\begin{proof}
We assume wlog that, for each $(k,\{ i,j\})\in M^*$, $cost(k,\{i,j\})=w(d_k, s_i)+u_{ij} $. We have:
\begin{equation}
   \begin{split}
   \label{equation404}
5 cost& (\text{CA($1, u$)}) \le 4 cost(\text{MA($1, u$)}) +cost(\text{TA($1$)})  \\
&\le  4( v_1(M_1)+ v_2(M_2))+ v_3(M_3) \ \  \text{(using~(\ref{equation401}) and~(\ref{equation402}))} \\
&\leq  \sum_{(k,\{ i,j\})\in M^*} (4 w(d_k, s_i)+ 3 w(d_k, s_j) + 4 w(s_i, s_j)) \ \    \text{(by Lemma's~\ref{lemma:ma_min} and~\ref{lemma:ta_min})} \\
&\leq  \sum_{(k,\{ i,j\})\in M^*} (7 w(d_k, s_i)+ 7 w(s_i, s_j)) \ \    \text{(by the triangle inequality)} \\
&=  \sum_{(k,\{ i,j\})\in M^*} 7 cost(k,\{i,j\}) \ \    \text{(by $cost(k,\{i,j\})=w(d_k, s_i)+u_{ij} $)} \\
&=   7 \ cost(M^*)
  \end{split}
\end{equation}

We now provide an instance for which this ratio is achieved.
Consider the instance $I$ depicted in~\cref{fig404}. If two points are not connected by an edge, their distance equals $ 5$.
Observe that an optimal solution is $ \{(k_1, \{1, 3\}), (k_2, \{2, 4\}),(k_3, \{5, 6\}), (k_4, \{7, 8\})\}$ with cost(M$^*$(I))=10. Note that $ M^*_R=\{\{1,3\}, \{2,4\}, \{5,6\}, \{7,8\}\}$. Let us now analyse the performance of MA($1, u$) and TA($1$) on instance $I$.
\begin{figure}
\centering
\tikz \draw     (4,0.9) node[circle]{$\{s_1, s_2\}$} (0.5,0.9) node[circle]{$\{k_1,s_3\}$} (7.5,0.9) node[circle]{$\{k_2,s_4\}$} (2.2,0.6) node[circle]{4} (5.8,0.6) node[circle]{4} (0.5,0.4) node[circle, fill=gray!45]{\ \ } -- (4,0.4) node[circle, fill=gray!45]{\ \ }--(7.5, 0.4) node[circle, fill=gray!45]{\ \ }-- cycle   (0.5,-1.4) node[circle]{$\{s_5, s_6\}$} (4,-1.4) node[circle]{$\{k_3,k_4\}$} (7.5,-1.4) node[circle]{$\{s_7, s_8\}$} (2.2,-0.6) node[circle]{1} (5.8,-0.6) node[circle]{1} (0.5,-0.9) node[circle, fill=gray!45]{\ \ } -- (4,-0.9) node[circle, fill=gray!45]{\ \ }--(7.5, -0.9) node[circle, fill=gray!45]{\ \ };
\caption{A worst-case instance for the combined algorithm CA($ 1, u$) of $\CS_{sum, s=t}$.}
\label{fig404}
\end{figure}

Based on the $u_{ij}$ values as defined in (\ref{equation101_u_1}), MA($1, u$) can find, in the first step, matching  $M_1=\{\{1,2\}, \{3,4\},  \{5,6\}, \{7,8\}\}$. Then, no matter how the second step matches the pairs to cars (since two cars stay at the same location), the total cost of MA($1, u$) will be $ 14$.

TA($1$) can assign request $1$ to car $\gamma(1)$, and request $3$ to car $\delta(1)$, and similarly, request $2$ to car $\gamma(2)$, and request $4$ to car $\delta(2)$; request $5$ to car $\gamma(3)$, and request $7$ to car $\delta(3)$; request $4$ to car $\gamma(4)$, and request $6$ to car $\delta(4)$. Note that $v_3(\{(k_3,5), (k_3,7), (k_4,6),(k_4,8)\})=v_3(\{(k_3,5), (k_3,6), (k_4,7),(k_4,8)\})=6$. Thus the total cost of TA($1$) is  $14$.

To summarize, the instance in~\cref{fig404} is a worst-case instance for the combined algorithm CA$(1,u)$. 
\end{proof}

\subsection{Approximation results for $\CS_{lat}$}
\label{subsec:approximation_results_CS_lat}
For $CS_{lat}$, we analyze algorithm CA($2, \mu$), which outputs the best of the two solutions, MA($2, \mu$) and TA($2$).

\begin{lemma}
\label{lemma:ca2}
	For each  $(k,\{i,j\})\in D\times R^2$,
$ 2w(d_k, s_i)+2w(d_k, s_j)+\mu_{ij}+\mu_{ji}+ \min \{2 w(d_k, s_i, t_i)+w(t_i, d_k, s_j, t_j), 2 w(d_k, s_j, t_j)+w(t_j, d_k, s_i, t_i)\} \le \min\{8w(d_k,s_i)+5\mu_{ij}, 8w(d_k,s_j)+5\mu_{ji}\}$.\\
\end{lemma}
The proof can be found in~\cref{A2}.

 \begin{theorem}
\label{the:alg_ana_2}
$\CA(2, \mu)$ is a 5/3-approximation algorithm for $\CS_{lat}$.  
\end{theorem}

\begin{proof}

 \begin{equation*}
   \begin{split}
  \label{equation403_11}
3 wait&(\text{CA($2, \mu$)}) \le  2 wait(\text{MA($2, \mu$)}) + wait(\text{TA($2$)}) \\ 
&\le  2 (v_1(M_1)+ v_2(M_2))+ v_3(M_3) \  \ \text{(by~\cref{lemma:ma} and~\cref{lemma:ta})}\\ 
&\leq  \sum_{(k,\{ i,j\})\in M^*} \min\{8w(d_k,s_i)+5\mu_{ij}, 8w(d_k,s_j)+5\mu_{ji}\}    \\
 &\leq \sum_{(k,\{i,j\})\in M^*} 5 \min\{2w(d_k, s_i)+\mu_{ij}, 2w(d_k, s_j)+\mu_{ji}\} \\
&=   5 \ wait(M^*).
  \end{split}
\end{equation*}
The third inequality follows from~\cref{lemma:ma_min},  \cref{lemma:ta_min} and~\cref{lemma:ca2}, i.e., $2 (v_1(M_1)+ v_2(M_2))+ v_3(M_3) \le  \sum_{(k,\{ i,j\})\in M^*} (2w(d_k, s_i)+2w(d_k, s_j)+\mu_{ij}+\mu_{ji}+ \min \{2 w(d_k, s_i, t_i)+w(t_i, d_k, s_j, t_j), 2 w(d_k, s_j, t_j)+w(t_j, d_k, s_i, t_i)\}) \le \sum_{(k,\{ i,j\})\in M^*} \min\{8w(d_k,s_i)+5\mu_{ij}, 8w(d_k,s_j)+5\mu_{ji}\} $. 
 \end{proof}

\subsection{Approximation results for $\CS_{lat,s=t}$}
\label{subsec:approximation_results_CS_latst}
Let us now consider the special case of $\CS_{lat}$ where the pick-up and drop-off location is identical to each request, $\CS_{lat,s=t}$. Note that  $ \mu_{ij}=\mu_{ji}=w(s_i,s_j)$.
Now, we show that for $\CS_{lat,s=t}$ algorithm CA$(2,\mu)$ is a $\frac{3}{2}$-approximation algorithm, a ratio which is strictly better than the ratio's of MA$(2,\mu)$ and TA(2) for this problem.


 \begin{theorem}
\label{the:Pl1_03}
$\CA(2, \mu)$ is a 3/2-approximation algorithm for $\CS_{lat,s=t}$. Moreover, there exists an instance $I$ for which  wait$(\CA(I))$=3/2 wait$(M^*(I))$.
\end{theorem}
\begin{proof}
We assume wlog that, for each $(k,\{ i,j\})\in M^*$, $wait(k,\{i,j\})=2w(d_k, s_i)+\mu_{ij} $. We have:
\begin{equation}
   \begin{split}
 \label{equation405}
2 wait&(\text{CA($2, \mu$)}) \le  wait(\text{MA($2, \mu$)}) + wait(\text{TA($2$)}) \\
&\le \sum_{(k,\{i,j\})\in \text{MA}} \min\{2w(d_k,  s_i) + \mu_{ij}, 2w(s_k, s_j)+\mu_{ji}\}\\
&\ \ \ +\sum_{(k,\{i,j\})\in \text{TA}} \min\{2w(d_k,  s_i) + \mu_{ij}, 2w(s_k, s_j)+\mu_{ji}\}
 \ \  \text{(by~(\ref{equation102_u_1}) and (\ref{eq:totallatency}))}  \\
&\le  v_1(M_1)+ v_2(M_2)+ v_3(M_3) \  \ \text{(by~\cref{lemma:ma} and~\cref{lemma:ta})}  \\
&\leq  \sum_{(k,\{ i,j\})\in M^*} (w(s_i,s_j)+ w(d_k, s_i)+ w(d_k, s_j) + 3 w(d_k, s_i)+ w(d_k, s_j))  \\
&\leq  \sum_{(k,\{ i,j\})\in M^*} (3w(s_i,s_j)+6 w(d_k, s_i))  \ \ \text{(by the triangle inequality)} \\
&=  \sum_{(k,\{ i,j\})\in M^*} 3 wait(k,\{i,j\}) \ \    \text{(by $wait(k,\{i,j\})=2w(d_k, s_i)+\mu_{ij} $)} \\
&=   3 \ wait(M^*).
  \end{split}
\end{equation}
The fourth inequality holds by~\cref{lemma:ma_min} and~\ref{lemma:ta_min}.

We now provide an instance for which this ratio is achieved.
Consider the instance $I$ depicted in Figure~(\ref{fig405}). If two points are not connected by an edge, their distance equals $ 5$. 
Observe that an optimal solution is $ \{(k_1, \{1, 3\}), (k_2, \{2, 4\}),(k_3, \{5, 6\}), (k_4, \{7, 8\})\}$ with wait(M$^*$(I))=8. Note that $ M^*_R=\{\{1,3\}, \{2,4\}, \{5,6\}, \{7,8\}\}$.
 Let us now analyse the performance of MA($2, \mu$) and TA(2) on instance $I$.

\begin{figure}
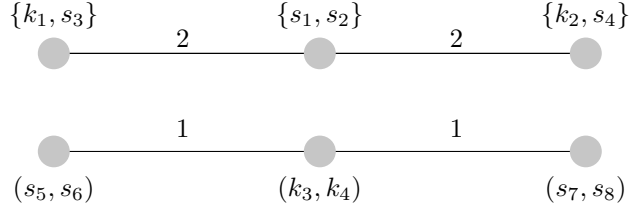

\centering
\tikz \draw     (4,0.9) node[circle]{$\{s_1, s_2\}$} (0.5,0.9) node[circle]{$\{k_1,s_3\}$} (7.5,0.9) node[circle]{$\{k_2,s_4\}$} (2.2,0.6) node[circle]{2} (5.8,0.6) node[circle]{2} (0.5,0.4) node[circle, fill=gray!45]{\ \ } -- (4,0.4) node[circle, fill=gray!45]{\ \ }--(7.5, 0.4) node[circle, fill=gray!45]{\ \ }-- cycle   (0.5,-1.4) node[circle]{$(s_5, s_6)$} (4,-1.4) node[circle]{$(k_3,k_4)$} (7.5,-1.4) node[circle]{$(s_7, s_8)$} (2.2,-0.6) node[circle]{1} (5.8,-0.6) node[circle]{1} (0.5,-0.9) node[circle, fill=gray!45]{\ \ } -- (4,-0.9) node[circle, fill=gray!45]{\ \ }--(7.5, -0.9) node[circle, fill=gray!45]{\ \ };
\caption{A worst-case instance for the combined algorithm CA($ 2, \mu$) of $CS_{lat, s=t}$.}
\label{fig405}
\end{figure}

Based on the $u_{ij}$ values as defined in (\ref{equation101_u_1}), MA($2, \mu$) can find, in the first step, matching $M_1=\{\{1,2\}, \{3,4\}, \{5,6\}, \{7,8\}\}$. Then, no matter how the second step matches the pairs to cars, the total cost of MA($2, \mu$) will be $12$.

TA($2$) can assign request $3$ to car $\gamma(1)$ and request $1$ to car $\delta(1)$, and similarly, request $4$ to car $\gamma(2)$ and request $2$ to car $\delta(2)$; request $5$ to car $\gamma(3)$ and request $7$ to car $\delta(3)$; request $4$ to car $\gamma(4)$ and request $6$ to car $\delta(4)$. Note that $v_3(\{(k_3,5), (k_3,7), (k_4,6),(k_4,8)\})=v_1(\{(k_3,5), (k_3,6), (k_4,7),(k_4,8)\})=8$. Thus the total cost of TA($2$) is  $12$.

To summarize, the instance in~\cref{fig405} is a worst-case instance for the combined algorithm CA$(2, \mu)$. 
\end{proof}

\section{Conclusions}
\label{sec:conclusions}

We have analyzed two algorithms for four different versions of a car sharing problem. One algorithm, called match and assign, first matches the requests into pairs, and then assigns the pairs to the cars. Another algorithm, called transportation assigns two requests to each car. These two algorithms emphasize different ingredients of the total cost and the total latency. Accordingly, we have proved that (for most problem variants) the worst-case ratio of the algorithm defined by the best of the two corresponding solutions is strictly better than the worst-case ratios of the individual algorithms.

We point out that the algorithms can be generalized to handle a variety of situations. We now list these situations, and shortly comment on the corresponding worst-case behavior.
 
\begin{description}
\item[Generalized car-sharing problem:$|R|=a \cdot n$.] In this situation, $a \cdot n$ requests are given, and each car can serve $a$ requests, $a \geq 2$.
We can extend algorithm TA for the resulting problem by replacing a single car by $a$ cars, and use an appropriately defined cost between a request in $R$ and a car. We obtain that the extended TA algorithm is a $(2a-1)$-approximation for $CS_{sum}$ and an $a$-approximation for $CS_{lat}$, generalizing our results for $a=2$. We refer to~\cref{subsec:GTAalgorithms} for the proofs.

\item[Related car-sharing problem: different speed.]
In this situation, we allow that cars have different speeds.  
Indeed, let car $k$ have speed $ p_k$ for each $k\in D$. 
We denote the travel time of serving requests in $R_k$ by $cost(k, R_k)/ p_k$ and the travel time of an allocation $M$ by $\sum_{(k, R_k)\in M} cost(k, R_k)/p_k$.
Analogously, we can adapt the latency of serving requests in $ R_k$ and the latency of an allocation $M$. 
Although it is unclear how to generalize algorithm MA, we can still use algorithm TA($\alpha$) in this situation. By defining $v_3(k,i) $ in terms of the cost above, we get the worst-case ratios of TA$(\alpha)$ as shown in~\cref{tab1}.

\item[Car redundancy or deficiency: $2|D| > |R|$ or $ 2|D| < |R|$.]
We shortly sketch how to modify
 TA($ \alpha$) for this situation.
For the problem with $2|D| > |R|$, by adding a number of dummy requests $R_d$ with $ |R_d|=2n-|R|$, where the distance between any two requests in $R_d$ is $0$, the distance between a request in $R_d$ and $R$  is a large constant, and the distance between a request in $ R_d$ and a car in $ D$ is $0$, an instance of our problem arises.
For the problem with $2|D|< |R|$, by adding a number of dummy cars $ D_d$ with $ |D_d|=\lceil |R|/2 \rceil-n$, where the distance between a car in $ D_d$ and a request in $ R$  is $0$, an instance of our problem arises.
We claim that minor modifications of the proofs for  TA imply that TA(1) is a 3-approximation algorithm for CS$_{sum}$, and TA(2) is a 2-approximation algorithm for CS$_{lat}$,
 On the other hand, it is not clear how to generalize  MA($\alpha, v$) for this situation.

\end{description}

\bibliographystyle{plainurl}
\bibliography{main}

\begin{thebibliography}{10}

\bibitem{ube}
How uberpool works.
\newblock \url{https://www.uber.com/nl/en/ride/uberpool/}.
\newblock Accessed: 2020-01-20.

\bibitem{Trans}
Transvision.
\newblock \url{https://www.transvision.nl/}.
\newblock Accessed: 2020-01-23.

\bibitem{agatz2011time}
Niels Agatz, Ann Campbell, Moritz Fleischmann, and Martin Savelsbergh.
\newblock Time slot management in attended home delivery.
\newblock {\em Transportation Science}, 45(3):435--449, 2011.

\bibitem{ashlagi2019edge}
Itai Ashlagi, Maximilien Burq, Chinmoy Dutta, Patrick Jaillet, Amin Saberi, and
  Chris Sholley.
\newblock Edge weighted online windowed matching.
\newblock In {\em Proceedings of the 2019 ACM Conference on Economics and
  Computation}, pages 729--742, 2019.

\bibitem{bei2018algorithms}
Xiaohui Bei and Shengyu Zhang.
\newblock Algorithms for trip-vehicle assignment in ride-sharing.
\newblock In {\em Thirty-Second AAAI Conference on Artificial Intelligence},
  2018.

\bibitem{gabow1990data}
Harold~N Gabow.
\newblock Data structures for weighted matching and nearest common ancestors
  with linking.
\newblock In {\em Proceedings of the first annual ACM-SIAM symposium on
  Discrete algorithms}, pages 434--443, 1990.

\bibitem{goossens2012between}
Dries Goossens, Sergey Polyakovskiy, Frits~CR Spieksma, and Gerhard~J
  Woeginger.
\newblock Between a rock and a hard place: the two-to-one assignment problem.
\newblock {\em Mathematical methods of operations research}, 76(2):223--237,
  2012.

\bibitem{stiglic2016making}
Mitja Stiglic, Niels Agatz, Martin Savelsbergh, and Mirko Gradisar.
\newblock Making dynamic ride-sharing work: The impact of driver and rider
  flexibility.
\newblock {\em Transportation Research Part E: Logistics and Transportation
  Review}, 91:190--207, 2016.

\bibitem{wang2018stable}
Xing Wang, Niels Agatz, and Alan Erera.
\newblock Stable matching for dynamic ride-sharing systems.
\newblock {\em Transportation Science}, 52(4):850--867, 2018.

\bibitem{williamson2011design}
David~P Williamson and David~B Shmoys.
\newblock {\em The design of approximation algorithms}.
\newblock Cambridge university press, 2011.

\end{thebibliography}

\newpage
\appendix
\section{Omitted Proofs}\label{app:omitted-proofs}
This section contains proofs omitted in the main body. 

\subsection{Proof of the claim $u_{ji}\le 2u_{ij}$ in~\cref{lem:alg_ana_401}}\label{A0_1}
If $u_{ji}\le u_{ij}$, the claim is proved;
otherwise, we prove the claim by distinguishing three cases based on $ u_{ij}$.
  If  $ u_{ij}=w(s_i,s_j,t_i, t_j)$, then, according to~(\ref{equation101_u}), $ u_{ji}\le w(s_j,s_i)+w(s_i,t_i)+w(t_i, t_j)$, and we have $ u_{ji}\le 2u_{ij}$;
 if $ u_{ij}=w(s_i,s_j,t_j, t_i)$, then, according to~(\ref{equation101_u}),  $ u_{ji}\le w(s_j,s_i)+w(s_i,t_j)+w(t_j, t_i)$, and we have $ u_{ji}\le 2u_{ij}$; otherwise, $ u_{ij}=w(s_i,t_i,s_j, t_j)$, and according to~(\ref{equation101_u}), $ u_{ji}\le w(s_j,t_j)+w(s_j,s_i)+w(s_i, t_i)$, so we have $ u_{ji}\le  2u_{ij}$.


\subsection{Proof of the claim $w(d_k,s_i, t_i, d_k,s_j, t_j) \le 3 cost(k, \{i,j\}) $ in~\cref{lem:alg_ana_402}}\label{A1}

According to~~(\ref{equation101_u_1}), $3 cost(k, \{i,j\})= 3w(d_k,s_i) +3u_{ij}$. The claim $3 cost(k, \{i,j\})\ge w(d_k,s_i, t_i,d_k, \\ s_j, t_j)$ is proved by distinguishing three cases based on $ u_{ij}$:

If $ u_{ij}=w(s_i,s_j,t_i, t_j)$, the inequality $3 cost(k, \{i,j\})\ge w(d_k,s_i)+ (w(s_i, s_j)+w(s_j, t_i)) +(w(k, s_i) +w(s_i, s_j)+w(s_j, t_i))+ (w(d_k,s_i)+w(s_i,s_j))+(w(s_j, t_i)+w(t_i, t_j)) \ge w(d_k,s_i, t_i, d_k, \\ s_j, t_j)$ follows by applying the triangle inequality.

If $ u_{ij}=w(s_i,s_j,t_j, t_i)$, the inequality $3 cost(k, \{i,j\})\ge w(d_k,s_i)+ (w(s_i, s_j)+(s_j, t_j)+w(t_j, t_i)) +(w(d_k, s_i) +w(s_i, s_j)+w(s_j, t_j)+w(t_j, t_i))+ (w(d_k,s_i)+w(s_i,s_j))+w(s_j, t_j) \ge w(d_k,s_i, t_i, d_k,s_j, t_j)$ follows by applying the triangle inequality.

If $ u_{ij}=w(s_i,t_i,s_j, t_j)$, the inequality $3 cost(k, \{i,j\})\ge w(d_k,s_i)+ w(s_i, t_i) +(w(d_k, s_i) +w(s_i, t_i))+ (w(d_k,s_i)+w(s_i,s_j))+w(s_j, t_j) \ge w(d_k,s_i, t_i, d_k,  s_j, t_j)$ follows by applying the triangle inequality.

All three cases entail the inequality $w(d_k,s_i, t_i, d_k,s_j, t_j) \le 3 cost(k, \{i,j\}) $.

\subsection{Proof of~\cref{lemma:ca2}}\label{A2}
We first prove $2w(d_k, s_i)+2w(d_k, s_j)+\mu_{ij}+\mu_{ji}+\min \{2 w(d_k, s_i, t_i)+w(t_i, d_k, s_j, t_j),  2 w(d_k, s_j, t_j)+w(t_j, d_k, s_i, t_i)\}\le 8w(d_k,s_i)+5\mu_{ij}$.
We distinguish three cases based on $ \mu_{ij}$:

\textit{Case 1} $ \mu_{ij}=w(s_i, t_i)+w(s_i, t_i, s_j,t_j)$. We have $ 2w(d_k, s_i)+\mu_{ij}= 2w(d_k, s_i)+2 w(s_i, t_i)+w(t_i, s_j)+w(s_j,t_j)$. Since $ w(d_k, t_i)\le w(d_k, s_i)+w(s_i, t_i)$, $w(d_k, s_j)\le(d_k, s_i)+w(s_i, t_i)+w(t_i, s_j) $, $u_{ij}=2 w(s_i, t_i)+w(t_i, s_j)+w(s_j,t_j) $ and  $ u_{ji}\le 2 w(s_j, t_j)+w(t_j, s_i)+w(s_i,t_i) \le 3 w(s_j, t_j)+2w(s_i, t_i)+w(t_i,s_j)  $, we have $2w(d_k, s_i)+2w(s_k, s_j)+\mu_{ij}+\mu_{ji}+2 w(d_k, s_i, t_i)+w(t_i, d_k, s_j, t_j)\le 8w(d_k, s_i)+10w(s_i, t_i)+5w(t_i, s_j)+ 5w(s_j, t_j) \le 8 w(d_k, s_i)+5\mu_{ij}$.

\textit{Case 2} $ \mu_{ij}=w(s_i, s_j,t_i)+ w(s_i, s_j,t_i, t_j)$.  We have  $ 2w(d_k, s_i)+\mu_{ij}= 2w(d_k, s_i)+2 w(s_i, s_j)+2w(s_j, t_i)+w(t_i, t_j)$. Since $w(s_i,t_i)\le w(s_i, s_j)+w(s_j, t_i)$, $ w(d_k, t_i)\le w(d_k, s_i)+w(s_i, s_j)+w(s_j, t_i)$,  $w(d_k, s_j)\le (d_k, s_i)+w(s_i, s_j) $, $w(s_j,t_j)\le w(s_j, t_i)+w(t_i, t_j)$, $u_{ij}=2 w(s_i, s_j)+2w(s_j, t_i)+w(t_i,t_j) $ and  $ u_{ji}\le 2 w(s_j, t_j)+w(t_j, s_i)+w(s_i,t_i) \le 4w(s_j, t_i)+2w(s_i,s_j)+3w(t_i, t_j) $, we have $2w(d_k, s_i)+2w(s_k, s_j)+\mu_{ij}+\mu_{ji}+2 w(d_k, s_i, t_i)+w(t_i, d_k, s_j, t_j)\le 8w(d_k, s_i)+10w(s_i, s_j)+10w(s_j, t_i)+ 5w(t_i, t_j) \le 8w(d_k, s_i)+5\mu_{ij}$.

\textit{Case 3} $ \mu_{ij}=w(s_i, s_j,t_j)+ w(s_i, s_j,t_j, t_i)$.  We have $ 2w(d_k, s_i)+\mu_{ij}= 2w(d_k, s_i)+2 w(s_i, s_j)+2w(s_j, t_j)+w(t_i, t_j)$. Since $w(d_k, s_j)\le w(d_k, s_i)+w(s_i, s_j)$, $w(d_k, t_j)\le w(d_k, s_i)+w(s_i, s_j)+w(s_j, t_j)$, and $w(s_i, t_i)\le (s_i, s_j)+w(s_j, t_j)+w(t_j, t_i) $, $u_{ij}=2 w(s_i, s_j)+2w(s_j, t_j)+w(t_i,t_j) $ and  $ u_{ji}\le 2 w(s_j, t_j)+w(t_j, s_i)+w(s_i,t_i)  \le 4w(s_j, t_i)+2w(s_i,s_j)+w(t_i, t_j) $, we have $2w(d_k, s_i)+2w(s_k, s_j)+\mu_{ij}+\mu_{ji}+2 w(d_k, s_j, t_j)+w(t_j, d_k, s_i, t_i)\le 8w(d_k, s_i)+10w(s_i, s_j)+10w(s_j, t_j)+ 5w(t_j, t_i)\le 8w(d_k, s_i)+5\mu_{ij}$.

Analogously, we have $2w(d_k, s_i)+2w(s_k, s_j)+\mu_{ij}+\mu_{ji}+\min \{2 w(d_k, s_i, t_i)+w(t_i, d_k, s_j, t_j), \\2 w(d_k, s_j, t_j)+w(t_j, d_k, s_i, t_i)\}\le 8w(d_k,s_j)+5\mu_{ji}$.

\section{Transportation algorithm for the generalized car-sharing problem} \label{subsec:GTAalgorithms}
Here, we present the transportation algorithm for the generalized car-sharing problem, in which each car serves $a$ ($ a>2$) requests.
In the transportation algorithm, we replace each car $ k\in D$ by  $a$ virtual cars $\gamma_1(k), \ldots, \gamma_a(k)$, resulting in car sets $\Gamma_1 =\{\gamma_1(1), ..., \gamma_1(n)\} $, ..., $\Gamma_a =\{\gamma_a(1), ..., \gamma_a(n)\} $. Next we assign requests to the $a\cdot n$ cars using a particular definition of the costs. 

By extending the definition of the travel time in~(\ref{equation101}), let $cost(k, R_k)$ denote the travel time needed to serve requests in $R_k$ and $cost(M)=\sum_{k=1}^{n} cost(k, R_k)$ denote the travel time of allocation $ M=\{(k, R_k): k\in D, R_k\in R\} $, in which
$cost(k, R_k)=cost(k, \{i_1, i_2, ..., i_a\})$ is the minimum travel time of visiting all locations $\{d_k, s_{i_1}, t_{i_1}, ...,  s_{i_a}, t_{i_a}\}$ where $s_{i_h}$ is visited before $ t_{i_h}$ for all $ i_h\in R_k$.
Similarly, by extending the definition of waiting time in~(\ref{equation102}), let $wait(k, R_k)$ denote the latency of serving requests in $R_k$ and $wait(M)=\sum_{k=1}^{n} wait(k, R_k)$ denote the latency of allocation $ M=\{(k, R_k): k\in D, R_k\in R\} $ where $ |R_k|=a$.

\subsection{TA$(a)_{sum}$ for $\CS_{sum}$}

 \begin{algorithm}	
	\caption{Transportation algorithm for $\CS_{sum}$ (TA$(a)_{sum}$)}
	\label{alg:GTA}
		\begin{algorithmic}[1]
	\STATE \emph{\normalsize Input}: non-negative weighted graph $G=(V,E,w)$, requests $R=\{i=(s_i,t_i): 1\le i\le m, s_i, t_i\in V\}$, cars $ D=\{k: 1\le k \le n, d_k\in V\}$, virtual car sets $\Gamma_1$, ..., $ \Gamma_a$.\\
\STATE \emph{\normalsize Output}: An allocation \emph{TA}$=\{(k,R_k):k\in D, R_k\in R\} $.\\
\STATE For $k=1,2,..., n-1$ do\\
\STATE \ \ \ For $1\le  j\le a, i \in R$ do\\
\STATE \ \ \ \ \ \ \ \  $v_3(\gamma_j(k),i)=w(d_k, s_i, t_i)+ w(t_i,d_k)$\\
\STATE \ \ \ end for\\
\STATE end for\\
 \STATE For $1\le  j\le a, i \in R$ do\\
\STATE \ \ \ \ \ \ \ \  $v_3(\gamma_j(n),i)=w(d_n, s_i, t_i)$\\
\STATE end for\\
\STATE Let $G_3\equiv (\bigcup_{j=1}^{a} \Gamma_j \cup R, v_3)$ be the complete bipartite graph with \emph{left} vertex-set $\bigcup_{j=1}^{a}\Gamma_j$, \emph{right} vertex-set $R$, and edges $v_3(x,i)$ for $ x\in \bigcup_{j=1}^{a}\Gamma_j$ and $i\in R $.\\
\STATE Find a minimum weight perfect  matching $M_3$ in  $G_3\equiv (\bigcup_{j=1}^{a}\Gamma_j \cup R, v_3)$ with weight $ v_3(M_3)$.\\
\STATE  Output allocation \emph{TA}$=M_4\equiv \{(k,\{i_1, ..., i_a\}):  (\gamma_1(k),i_1),  ..., (\gamma_a(k), i_a) \in M_3 \}$.
	\end{algorithmic}		
\end{algorithm}

We refer to the solution found by the extended transportation algorithm for the general car-sharing problem with minimum travel time as TA$(a)_{sum}$. For each $(k,\{i_1, ..., i_a\})\in $TA$(a)_{sum}$,
we order the requests $\{i_1,  ..., i_a\}$ such that $w(d_k, t_{i_j})\ge w(d_k, t_{i_{j-1}})$ for all $ 1< j\le a$.
The crucial points of algorithm TA$(a)_{sum}$ are found in lines 3 to 10 in~\cref{alg:GTA} where a travel time is defined; a resulting quantity is $v_3(M_3)$. For every request set $\{i_1,i_2, ..., i_a\}$ that are assigned to a car $k$, $v_3(M_3)$ chooses the path $d_k\rightarrow s_{i_1} \rightarrow t_{i_1} \rightarrow d_k \rightarrow s_{i_2} \rightarrow t_{i_2}\rightarrow .. \rightarrow d_k \rightarrow s_{i_a} \rightarrow t_{i_a}$.

Using $cost(M)=\sum_{k=1}^{n} cost(k, R_k)$ and $ cost (k,\{i_1, i_2, ..., i_a\})\le   \sum_{j=1}^{a}  w(d_k, s_{i_j}, t_{i_j}, d_k) - \max_{1 \le i_h \le a} w(d_k, t_{i_h})$, we have:
\begin{equation}
   \label{claim:gta1}
\text{cost(TA$(a)_{sum}$)$\le  v_3(M_3)$.}
\end{equation}

We now prove two lemma's concerning this quantity $v_3(M_3)$, which will be of use in~\cref{lem:alg_geta_sum}.
\begin{lemma}
\label{lemma:gta_sum}
For any $a\ge 2$, we have:
\[ v_3(M_3)= \sum_{(k,\{i_1, i_2, ..., i_a\})\in M_4} ( \sum_{j=1}^{a}  w(d_k, s_{i_j}, t_{i_j}, d_k) - \max_{1 \le i_h \le a} w(d_k, t_{i_j})).\]
\end{lemma}
\begin{proof}
We claim that $ v_3(M_3)$ is minimized if and only if for each car $ k\in D$, car $k $ serves request $i_a\in R_k $ after serving all serves requests in $R_k\setminus  \{i_a\}=\{i_1, i_2, ..., i_{a-1}\}$. If this claim holds, then based on line 3 to line 10 in~\cref{alg:GTA}, $ v_3(k, \{i_1, i_2,..., i_a\})=  \sum_{j \neq h}  w(d_k, s_{i_j}, t_{i_j}, d_k) - \max_{1 \le i_h \le a} w(d_k, t_{i_j})$ and thus 
\[v_3(M_3)= \sum_{(k,\{i_1, i_2, ..., i_a\})\in M_4} (  \sum_{j=1}^{a}   w(d_k, s_{i_j}, t_{i_j}, d_k) - \max_{1 \le i_h \le a} w(d_k, t_{i_h})).\] 

It remains to prove the claim. Consider any car $ k\in D$ and $ i_x \in R_k \setminus  \{i_a\}$. We will prove that $ v_3(M_3)$ is minimized when serving request  $i_x $ earlier than request $ i_a$  if and only if $  w(d_k, t_{i_a})> w(d_k, t_{i_x})$.

\emph{Necessary condition} If request  $i_x \in R_k \setminus  \{i_a\}$ is served earlier than request $ i_a$, then, when minimizing $ v_3(k, R_k)$, it is true that $  w(d_k, t_{i_a})> w(d_k, t_{i_x})$ holds.  If $ w(d_k, s_{i_x}, t_{i_x}) + w(d_k, t_{i_x}) +  w(d_k, s_{i_a}, t_{i_a}) \le  w(d_k, s_{i_a}, t_{i_a}) +  w(d_k, t_{i_a}) +  w(d_k, s_{i_x}, t_{i_x})$, then $  w(d_k, t_{i_x})\le  w(d_k, t_{i_a})$.

\emph{Sufficient condition} If  $  w(d_k, t_{i_a})> w(d_k, t_{i_x})$, then request  $i_x $ must be served earlier than request $ i_a$ for minimizing $ v_3(k, \{i_1, i_2,..., i_a\})$. If $  w(d_k, t_{i_a})\ge  w(d_k, t_{i_x})$, then $ w(d_k, s_{i_a}, t_{i_a}) +  w(d_k, t_{i_a}) +  w(d_k, s_{i_x}, t_{i_x}) \ge  w(d_k, s_{i_x}, t_{i_x}) + w(d_k, t_{i_x}) +  w(d_k, s_{i_a}, t_{i_a})$,  that means  $ v_3(M_3)$ is smaller when serving request  $i_x $ earlier than request $ i_a$.
\end{proof}

\begin{lemma}
\label{lemma:gta_sum}
For $a\ge 2$ and for each allocation $ M$, we have:
\[ v_3(M_3)\le \sum_{(k,\{i_1,i_2, ..., i_a\})\in M} ( \sum_{j=1}^{a}  w(d_k, s_{i_j}, t_{i_j}, d_k) - \max_{1 \le i_h \le a} w(d_k, t_{i_h})).\]
\end{lemma}

This lemma is obvious since $M_3$ is a minimum weight perfect matching. 

\begin{lemma}
\label{lem:alg_geta_sum}
$\TA_{sum}(a)$ is a $(2a-1)$-approximation algorithm for $CS_{sum}$. 
\end{lemma}
\begin{proof}
\begin{equation*}
  \begin{split}
   \label{equation_gta_sum}
&cost(\text{TA$(a)_{sum}$})\le  v_3(M_3)  \  \ \text{(by~(\ref{claim:gta1}))}\\
 &\leq \sum_{(k,\{i_1,i_2, ..., i_a\})\in M^*} ( \sum_{j=1}^{a}  w(d_k, s_{i_j}, t_{i_j}, d_k) - \max_{1 \le i_h \le a} w(d_k, t_{i_h})) \ \    \text{(by~\cref{lemma:gta_sum})} \\
  &= \sum_{(k,\{i_1,i_2, ..., i_a\})\in M^*} ( \sum_{j=1}^{a}  w(d_k, s_{i_j}, t_{i_j})+\sum_{j=1}^{a}  w(t_{i_j}, d_k) - \max_{1 \le i_h \le a} w(d_k, t_{i_h}))  \\
  &\leq \sum_{(k,R_k)\in M^*} (2a-1) cost(k, R_k) \ \    \text{(by $cost(k, R_k)\le \max_{i_h\in R_k} w(d_k, s_{i_h}, t_{i_h})$)} \\
   &= (2a-1) \ cost(M^*)
     \end{split}
\end{equation*}
\end{proof}

\subsection{TA$(a)_{lat}$ for $\CS_{lat}$}

 \begin{algorithm}	
	\caption{Transportation algorithm for $\CS_{lat}$ (TA$(a)_{lat}$)}
	\label{alg:GTA_lat}
		\begin{algorithmic}[1]
	\STATE \emph{\normalsize Input}: non-negative weighted graph $G=(V,E,w)$, requests $R=\{i=(s_i,t_i): 1\le i\le m, s_i, t_i\in V\}$, cars $ D=\{k: 1\le k \le n, d_k\in V\}$, dummy car sets $D_j=\{k: k\in D\}$ for all $ 1\le j\le a$.\\
\STATE \emph{\normalsize Output}: An allocation \emph{TA}$=\{(k,R_k):k\in D, R_k\in R\} $.\\
\STATE For $k=1,2,..., n$ do\\
\STATE For $1\le  j\le a, i \in R$ do\\
\STATE \ \ \ \ \ \ \ \  $v_3(\gamma_j(k),i)=(a-j+1)\cdot w(d_k, s_i, t_i)+(a-j)\cdot w(t_i,d_k)$\\
\STATE end for\\
\STATE end for\\
\STATE Let $G_3\equiv (\bigcup_{j=1}^{a} \Gamma_j \cup R, v_3)$ be the complete bi-partite graph with \emph{left} vertex-set $\bigcup_{j=1}^{a}\Gamma_j$, \emph{right} vertex-set $R$, and edges $v_3(x,i)$ for $ x\in \bigcup_{j=1}^{a}\Gamma_j$ and $i\in R $.\\
\STATE Find a minimum weight perfect  matching $M_3$ in  $G_3\equiv (\bigcup_{j=1}^{a}\Gamma_j \cup R, v_3)$ with weight $ v_3(M_3)$.\\
\STATE  Output allocation \emph{TA}$=M_4\equiv \{(k,\{i_1, ..., i_a\}): (\gamma_1(k),i_1), ..., (\gamma_a(k), i_a) \in M_3 \}$.
	\end{algorithmic}			
\end{algorithm}

We refer to the solution found by the extended transportation algorithm for the general car-sharing problem with minimum total latency as TA$(a)_{lat}$.
We count the waiting time for each customer depending to which dummy car the corresponding request was assigned.
For each $(k,\{i_1,  ..., i_a\})\in $TA$(a)_{lat}$, we order every request set $\{i_1, ..., i_a\}$ such that $w(d_k, s_{i_j}, t_{i_j})\ge w(d_k, s_{i_{j-1}}, t_{i_{j-1}})$ for all $ 1< j\le a$.
The crucial point of algorithm TA$(a)_{lat}$ is found in line 5, where a waiting time is defined; a resulting quantity is $v_3(M_3)$. For every request set $\{i_1,i_2, ..., i_a\}$ that is assigned to car $k$, $v_3(M_3)$ chooses the path $d_k\rightarrow s_{i_1} \rightarrow t_{i_1} \rightarrow d_k \rightarrow s_{i_2} \rightarrow t_{i_2}\rightarrow .. \rightarrow d_k \rightarrow s_{i_a} \rightarrow t_{i_a}$.

Using $wait(M)=\sum_{k=1}^{n} wait(k, R_k)$ and $ wait (k,\{i_1, i_2, ..., i_a\})\le  \sum_{j=1}^{a} ((a-j+1) \cdot w(d_k, s_{i_j}, t_{i_j}) + (a-j) \cdot  w(d_k, t_{i_j}))$, we have:
\begin{equation}
   \label{claim:gta2}
\text{wait(TA$(a)_{lat}$)$\le  v_3(M_3)$.}
\end{equation}

We now prove two lemma's concerning this quantity $v_3(M_3)$, which will be of use in~\cref{lem:alg_geta_lat}.
\begin{lemma}
\label{lemma:gta_lat}
For any $a\ge 2$, we have:
\[ v_3(M_3)= \sum_{(k,\{i_1, i_2, ..., i_a\})\in M_4}  \sum_{j=1}^{a} ( (a-j+1)  \cdot w(d_k, s_{i_j}, t_{i_j}) + (a-j) \cdot  w(d_k, t_{i_j})).\]
\end{lemma}
\begin{proof}
We claim that $ v_3(M_3)$ is minimized if and only if for each car $ k\in D$, car $k $ serves requests $R_k=\{i_1, i_2, ..., i_a\}$ ($(k,i_1), (k, i_2), ..., (k, i_a) \in M_3$) in order of non-decreasing travel time $w(d_k, s_{i_h}, t_{i_h}) +  w(d_k, t_{i_h})$ ($i_h\in R_k$). If this claim holds, then based on line 5 in~\cref{alg:GTA_lat}, $ v_3(k, \{i_1, i_2,..., i_a\})= \sum_{j=1}^{a} ( (a-j+1)  \cdot w(d_k, s_{i_j}, t_{i_j}) + (a-j) \cdot  w(d_k, t_{i_j}))$ and thus $ v_3(M_3)= \sum_{(k,\{i_1, i_2, ..., i_a\})\in M_4}  \sum_{j=1}^{a} ( (a-j+1)  \cdot w(d_k, s_{i_j}, t_{i_j}) + (a-j) \cdot  w(d_k, t_{i_j}))$.

It remains to prove the claim. Consider any car $ k\in D$ and $ i_x, i_y \in R_k$. We will prove that $ v_3(M_3)$ is minimized when serving request  $i_x$ earlier than request $ i_y$  if and only if $ w(d_k, s_{i_x}, t_{i_x}) +  w(d_k, t_{i_x})< w(d_k, s_{i_y}, t_{i_y}) + w(d_k, t_{i_y})$.

\emph{Necessary condition} If request  $i_x \in R_k$ is served earlier than request $ i_y\in R_k$, then, when minimizing $ v_3(k, R_k)$, it is true that $ w(d_k, s_{i_x}, t_{i_x}) +  w(d_k, t_{i_x})< w(d_k, s_{i_y}, t_{i_y}) + w(d_k, t_{i_y})$ holds.  If $ (a-x+1)  \cdot w(d_k, s_{i_x}, t_{i_x}) + (a-x)  \cdot  w(d_k, t_{i_x}) + (a-y+1)  \cdot w(d_k, s_{i_y}, t_{i_y}) + (a-y)  \cdot  w(d_k, t_{i_y}) \le  (a-y+1)  \cdot w(d_k, s_{i_x}, t_{i_x}) + (a-y)  \cdot  w(d_k, t_{i_x}) + (a-x+1)  \cdot w(d_k, s_{i_y}, t_{i_y}) + (a-x)  \cdot  w(d_k, t_{i_y})$ with $ x< y$, then $ (y-x)( w(d_k, s_{i_x}, t_{i_x}) +  w(d_k, t_{i_x})) \le (y-x)(w(d_k, s_{i_y}, t_{i_y}) + w(d_k, t_{i_y}))$, and thus we have $   w(d_k, s_{i_x}, t_{i_x}) +  w(d_k, t_{i_x})< w(d_k, s_{i_y}, t_{i_y}) + w(d_k, t_{i_y})$.

\emph{Sufficient condition} If  $ w(d_k, s_{i_x}, t_{i_x}) +  w(d_k, t_{i_x})<  w(d_k, s_{i_y}, t_{i_y}) + w(d_k, t_{i_y})$, then request  $i_x $ must be served earlier than request $ i_y$ when minimizing $ v_3(k, \{i_1, i_2,..., i_a\})$. If $   w(d_k, s_{i_x}, t_{i_x}) +  w(d_k, t_{i_x})< w(d_k, s_{i_y}, t_{i_y}) + w(d_k, t_{i_y})$ with $ a\le y>x>0 $, then $(y-x)\cdot (w(d_k, s_{i_x}, t_{i_x}) +  w(d_k, t_{i_x}))\le (y-x)\cdot  (w(d_k, s_{i_y}, t_{i_y}) + w(d_k, t_{i_y}))$, and thus $ (a-x+1)  \cdot w(d_k, s_{i_x}, t_{i_x}) + (a-x)  \cdot  w(d_k, t_{i_x}) + (a-y+1)  \cdot w(d_k, s_{i_y}, t_{i_y}) + (a-y)  \cdot  w(d_k, t_{i_y}) \le  (a-y+1)  \cdot w(d_k, s_{i_x}, t_{i_x}) + (a-y)  \cdot  w(d_k, t_{i_x}) + (a-x+1)  \cdot w(d_k, s_{i_y}, t_{i_y}) + (a-x)  \cdot  w(d_k, t_{i_y})$. This means that $v_3(M_3)$ is smaller when serving request  $i_x $ earlier than request $ i_y$.
\end{proof}

\begin{lemma}
\label{lemma:gta_lat}
For  $a\ge 2$ and for each allocation $ M=\{(k, \{i_1, i_2, ..., i_a\} ): k\in D,  \{i_1, i_2, ..., i_a\} \in R\} $, suppose $w(d_k, s_{i_j}, t_{i_j})\ge w(d_k, s_{i_{j-1}}, t_{i_{j-1}})$ for all $ 1< j\le a$, we have:
\[ v_3(M_3)\le \sum_{(k,\{i_1,i_2, ..., i_a\})\in M}  \sum_{j=1}^{a} ((a-j+1)\cdot w(d_k, s_{i_j}, t_{i_j}) + (a-j) \cdot  w(d_k, t_{i_j})). \]
\end{lemma}

This lemma is obvious since $M_3$ is a minimum weight perfect  matching.

\begin{lemma}
\label{lem:alg_geta_lat}
$\TA_{lat}(a)$ is an $a$-approximation algorithm for $\CS_{lat}$.
\end{lemma}
\begin{proof}
\begin{equation*}
  \begin{split}
   \label{equation_gta_lat}
&wait(\text{TA$_{lat}(a)$})\le  v_3(M_3)  \  \ \text{(by~(\ref{claim:gta2}))}\\
 &\leq \sum_{(k,\{i_1,i_2, ..., i_a\})\in M^*}  \sum_{j=1}^{a} ( (a-j+1) \cdot w(d_k, s_{i_j}, t_{i_j}) + (a-j) \cdot  w(d_k, t_{i_j})) \ \    \text{(by~\cref{lemma:gta_lat})} \\
 &= \sum_{(k,\{i_1,i_2, ..., i_a\})\in M^*} \sum_{j=1}^{a} (a\cdot w(d_k, s_{i_j}, t_{i_j})- (j-1) \cdot w(d_k, s_{i_j}, t_{i_j}) + (a-j) \cdot  w(d_k, t_{i_j}) ) \\
  &\leq \sum_{(k,\{i_1,i_2, ..., i_a\})\in M^*}  a \cdot \sum_{j=1}^{a} w(d_k, s_{i_j}, t_{i_j})  \\
    &\leq \sum_{(k,\{i_1,i_2, ..., i_a\})\in M^*} a\cdot wait (k, R_k) \ \    \text{(by $wait(k, R_k)\le \sum_{i_h\in R_k} w(d_k, s_{i_h}, t_{i_h})$)} \\
   &= a \cdot wait(M^*)
     \end{split}
\end{equation*}
The third inequality holds since we ordered the requests in $R_k $ ($k\in D$) with non-decreasing $ w(d_k, s_{i_j}, t_{i_j})$ for all $ j\in R_k$, and by triangle inequality.
\end{proof}

\section{Proofs for the rest results in~\cref{tab1}}\label{app:tableresults}

Recall the results in~\cref{tab1}. We will prove the rest of the results that were not proven in the main text.

\subsection{Remaining proofs of results for  $\CS_{sum}$}\label{A5_sum}
In $\CS_{sum}$, for each  $(k,\{i,j\})$, we have
\begin{eqnarray}\label{inequation1_sum} cost(k,\{i,j\})= \min\{w(d_k, s_i)+ u_{ij}, w(d_k, s_j)+ u_{ji}\}, \mbox{ and }\\
\label{inequation2_sum} cost(k,\{i,j\})\ge \max\{w(d_k, s_i, t_i), w(d_k, s_j, t_j), w(d_k, t_i), w(d_k, t_j)\}.\end{eqnarray}

 \begin{lemma}
\label{claim_uu}
For each request pair $\{i,j\}\in R$, we have
\[\frac{\mu_{ij}+\mu_{ji}}{2}+w(s_i, s_j)\le \min\{3u_{ij}, 3u_{ji}\}\] and \[\frac{\mu_{ij}+\mu_{ji}}{2}+w(s_i, s_j)\le \min\{2\mu_{ij}, 2\mu_{ji}\}.\]
\end{lemma}
\begin{proof}
We first prove $\frac{\mu_{ij}+\mu_{ji}}{2}+w(s_i, s_j)\le 3u_{ij}$ and $\frac{\mu_{ij}+\mu_{ji}}{2}+w(s_i, s_j)\le 2\mu_{ij}$ by distinguishing three cases.

\textit{Case 1} $ u_{ij}=w(s_i, t_i,s_j, t_j)$ (resp.  $\mu_{ij}= 2w(s_i, t_i)+w(t_i, s_j, t_j)$). Observe that $\mu_{ij}\le 2w(s_i, t_i)+w(t_i, s_j, t_j)$ (resp. $ u_{ij}\le w(s_i, t_i,s_j, t_j)$) and $\mu_{ji}\le 2w(s_j, t_j)+w(t_j, s_i, t_i)$. By the triangle inequality $ w(s_i, t_j)\le w(s_i, t_i, s_j, t_j) $ and $ w(s_i, s_j)\le w(s_i, t_i, s_j) $, we have $\frac{\mu_{ij}+\mu_{ji}}{2}+w(s_i, s_j)\le 3 w(s_i,t_i)+2w(t_i,s_j)+2w(s_j, t_j) \le 3u_{ij}$ (resp. $\frac{\mu_{ij}+\mu_{ji}}{2}+w(s_i, s_j)\le  3 w(s_i,t_i)+2w(t_i,s_j)+2w(s_j, t_j)\le 2\mu_{ij}$).

\textit{Case 2} $ u_{ij}=w(s_i,s_j,t_i, t_j)$ (resp. $\mu_{ij}= 2w(s_i, s_j, t_i)+w(t_i, t_j)$). Observe that $\mu_{ij}\le 2w(s_i, s_j, t_i)+w(t_i, t_j)$  (resp. $ u_{ij}\le w(s_i,s_j,t_i, t_j)$) and $\mu_{ji}\le 2w(s_j, t_j)+w(t_j, s_i, t_i)$. By the triangle inequality $ w(s_j, t_j)\le w(s_j, t_i, t_j) $ and $ w(s_i, t_i)\le w(s_i, s_j, t_i) $ and $ w(s_i, t_j)\le w(s_i, s_j, t_i, t_j) $, we have $\frac{\mu_{ij}+\mu_{ji}}{2}+w(s_i, s_j)\le 3 w(s_i,s_j)+3w(s_j,t_i)+w(t_i, t_j) \le 3u_{ij}$ (resp.  $\frac{\mu_{ij}+\mu_{ji}}{2}+w(s_i, s_j)\le 3 w(s_i,s_j)+3w(s_j,t_i)+w(t_i, t_j)\le 2\mu_{ij}$).

\textit{Case 3} $ u_{ij}=w(s_i,s_j, t_j, t_i)$ (resp. $\mu_{ij}= 2w(s_i, s_j, t_j)+w(t_j, t_i)$). Observe that $\mu_{ij}\le 2w(s_i, s_j, t_j)+w(t_j, t_i)$ (resp. $ u_{ij}\le w(s_i,s_j, t_j, t_i)$)  and $\mu_{ji}\le 2w(s_j, t_j)+w(t_j, s_i, t_i)$. By the triangle inequality $ w(s_i, t_j)\le w(s_i, s_j, t_j) $ and $ w(s_i, t_i)\le w(s_i, s_j, t_j, t_i) $, we have $\frac{\mu_{ij}+\mu_{ji}}{2}+w(s_i, s_j)\le 3 w(s_i,s_j)+3w(s_j,t_j)+w(t_j, t_i) \le 3u_{ij}$ (resp.  $\frac{\mu_{ij}+\mu_{ji}}{2}+w(s_i, s_j)\le 3 w(s_i,s_j)+3w(s_j,t_j)+w(t_j, t_i) \le 2\mu_{ij}$).

The proofs  of  $\frac{\mu_{ij}+\mu_{ji}}{2}+w(s_i, s_j)\le 3u_{ji}$ and $\frac{\mu_{ij}+\mu_{ji}}{2}+w(s_i, s_j)\le 2\mu_{ji}$ are symmetric, thus the claim is proved.
\end{proof}

\begin{lemma}
\label{lem:allapp_lem_1}
$\MA(2, \mu)$ is a $3$-approximation algorithm for $\CS_{sum}$.
\end{lemma}
\begin{proof}
\begin{equation}
   \begin{split}
 \label{equation_ma2_sum}
cost&(\text{MA$(2, \mu)$})=  \sum_{(k,\{ i,j\})\in \text{MA}} \min\{w(d_k, s_i)+u_{ij}, w(d_k, s_j)+u_{ji}\} \ \ \text{(by~(\ref{equation101_u_1}))}\\
&\leq  \sum_{(k,\{ i,j\})\in M^*} (2w(d_k, s_i)+2w(d_k, s_j)+\frac{\mu_{ij}+\mu_{ji}}{2})\\
&\leq  \sum_{(k,\{ i,j\})\in M^*} (2w(d_k, s_i)+ w(s_i, s_j)+\frac{\mu_{ij}+\mu_{ji}}{2})\ \ \text{(by the triangle inequality)}\\
&\leq  \sum_{(k,\{ i,j\})\in M^*} (2w(d_k, s_i)+3 \min\{u_{ij}, u_{ji}\})\ \ \text{(by~\cref{claim_uu})}\\
&\leq  \sum_{(k,\{ i,j\})\in M^*} 3 cost(k,\{i,j\})  \ \ \text{(by~(\ref{inequation1_sum}))}\\
&\leq   3 \ cost(M^*)
  \end{split}
\end{equation}
The first inequality follows from~\cref{lemma:ma_min}.
\end{proof}

\begin{lemma}
\label{lem:allapp_lem_2}
$\TA(2)$ is a $4$-approximation algorithm for $\CS_{sum}$.
\end{lemma}
\begin{proof}
\begin{equation}
   \begin{split}
 \label{equation_ta2_sum}
&cost(\text{TA$(2)$})=  \sum_{(k,\{ i,j\})\in \text{TA}} \min\{w(d_k, s_i)+u_{ij}, w(d_k, s_j)+u_{ji}\} \ \ \text{(by~(\ref{equation101_u_1}))}\\
&\leq  \sum_{(k,\{ i,j\})\in M^*} \min\{2 w(d_k, s_i, t_i)+ w(t_i, d_k, s_j, t_j), 2 w(d_k, s_j, t_j)+ w(t_j, d_k, s_i, t_i) \}  \\
&\leq  \sum_{(k,\{ i,j\})\in M^*} 4 cost(k,\{i,j\})  \ \ \text{(by~(\ref{inequation2_sum}))}\\
&\leq   4 \ cost(M^*)
  \end{split}
\end{equation}
The first inequality holds since TA is a minimum weight perfect matching in algorithm TA$(2)$.
\end{proof}

 \begin{theorem*}
\label{the:allapp_the_1}
$\CA(2, \mu)$ is a 3-approximation algorithm for $\CS_{sum}$. 
\end{theorem*}
According to~(\ref{equation_ma2_sum}) and~(\ref{equation_ta2_sum}), we have:  $\text{cost(CA}(I))=\min\{\text{cost(MA}(I)),  \text{cost(TA}(I))\} \le$ 3 cost($M^*$(I)) for any instance $I$.

\subsection{Remaining proofs of results for $\CS_{sum, s=t}$}\label{A5_sum_st}
In $\CS_{sum, s=t}$, for each $(k,\{i,j\})$, by~(\ref{equation101}) and the triangle inequality, we have
\begin{eqnarray}\label{inequation1_sum_s=t} \frac{3}{2}cost(k,\{i,j\})\ge \frac{1}{2}(w(d_k, s_i)+ w(d_k, s_j))+ w(s_i, s_j) \mbox{ and}, \\ \label{inequation2_sum_s=t} cost(k,\{i,j\})\ge \max\{w(d_k, s_i), w(d_k, s_j)\}.\end{eqnarray}

\begin{lemma*}
\label{lem:allapp_lem_3}
$\MA(1,u)$ is a $3/2$-approximation algorithm for $\CS_{sum,s=t}$. Moreover, there exists an instance $I$ for which  cost$(\MA(I))$=3/2 cost$(M^*(I))$.
\end{lemma*}
\begin{proof}
\begin{equation}
   \begin{split}
 \label{equation_ma1_sum_s=t}
cost(\text{MA$(1,u)$})&=\sum_{(k,\{i,j\})\in \text{MA}} \min\{w(d_k, s_i) +u_{ij}, w(d_k, s_j)+u_{ji}\}  \ \   \text{(by~(\ref{equation101_u_1}))} \\
&\leq  \sum_{(k,\{ i,j\})\in M^*} (w(s_i,s_j)+ \frac{1}{2}(w(d_k, s_i)+ w(d_k, s_j)))  \\
&\leq   \sum_{(k,\{ i,j\})\in M^*} 3/2 cost(k,\{i,j\})  \ \ \text{(by~(\ref{inequation1_sum_s=t}))} \\
&\leq   3/2 \ cost(M^*)
  \end{split}
\end{equation}
The first inequality holds  since MA is a minimum weight perfect matching in algorithm MA$(1,u)$ and $u_{ij}=u_{ji}=w(s_i,s_j) $.  To see that equality  may hold in (\ref{equation_ma1_sum_s=t}), consider the subgraph induced by the nodes $(k_1,s_3)$, $(s_1,s_2)$, and $(k_2,s_4)$ in \cref{fig404}.
\end{proof}

\begin{lemma*}
\label{lem:allapp_lem_4}
$\TA(1)$ is a $3$-approximation algorithm for $\CS_{sum,s=t}$.  Moreover, there exists an instance $I$ for which  cost$(\TA(I))$=3 cost$(M^*(I))$.
\end{lemma*}
\begin{proof}
\begin{equation}
   \begin{split}
 \label{equation_ta1_sum_s=t}
cost(\text{TA$(1)$})&=\sum_{(k,\{i,j\})\in \text{TA}} \min\{w(d_k, s_i) +u_{ij}, w(d_k, s_j)+u_{ji}\}  \ \ \text{(by~(\ref{equation101_u_1}))} \\
&\leq  \sum_{(k,\{ i,j\})\in TA} \min\{2 w(d_k, s_i)+ w(d_k, s_j), 2 w(d_k, s_j)+ w(d_k, s_i) \}  \\
&\leq  \sum_{(k,\{ i,j\})\in M^*} \min\{2 w(d_k, s_i)+ w(d_k, s_j), 2 w(d_k, s_j)+ w(d_k, s_i) \}  \\
&\leq  \sum_{(k,\{ i,j\})\in M^*} 3 cost(k,\{i,j\})  \ \ \text{(by~(\ref{inequation2_sum_s=t}))}\\
&\leq   3 \ cost(M^*)
  \end{split}
\end{equation}

The first inequality holds by $u_{ij}=u_{ji}=w(s_i,s_j)$ and the triangle inequality, and the second inequality holds since TA is a minimum weight perfect matching in algorithm TA$(1)$.
To see that equality may hold in (\ref{equation_ta1_sum_s=t}), consider the subgraph induced by the nodes $(s_5, s_6)$, $(k_3, k_4)$, and $(s_7,s_8)$ in~\cref{fig404}.
\end{proof}

\begin{lemma}
\label{lem:allapp_lem_5}
$\MA(2,\mu$) is a $3/2$-approximation algorithm for $\CS_{sum,s=t}$.
\end{lemma}
\begin{proof}
\begin{equation*}
   \begin{split}
 \label{equation_ma2_sum_s=t}
cost&(\text{MA($2,\mu$)})=\sum_{(k,\{i,j\})\in \text{MA}} \min\{w(d_k, s_i) +u_{ij}, w(d_k, s_j)+u_{ji}\}  \ \   \text{(by~(\ref{equation101_u_1}))} \\
&=  \sum_{(k,\{ i,j\})\in MA} (w(s_i,s_j)+ \min\{w(d_k, s_i), w(d_k, s_j)\}) \\
&\leq  \sum_{(k,\{ i,j\})\in M^*} (w(s_i,s_j)+  \frac{1}{2}(w(d_k, s_i)+ w(d_k, s_j))) \\
&\leq  3/2 \sum_{(k,\{ i,j\})\in M^*} cost(k,\{i,j\})  \ \ \text{(by~(\ref{inequation1_sum_s=t})) } \\
&\leq   3/2 \ cost(M^*)
  \end{split}
\end{equation*}
The first inequality holds because $u_{ij}=u_{ji}=w(s_i,s_j)$, and the second inequality inequality follows from~\cref{lemma:ma_min}.
\end{proof}

\begin{lemma}
\label{lem:allapp_lem_6}
$\TA(2)$ is a $4$-approximation algorithm for $\CS_{sum,s=t}$.
\end{lemma}
\begin{proof}
\begin{equation*}
   \begin{split}
 \label{equation_ta2_sum_s=t}
cost&(\text{TA($2$)})=  \sum_{(k,\{ i,j\})\in \text{TA}} \min\{w(d_k, s_i)+u_{ij}, w(d_k, s_j)+u_{ji}\} \ \ \text{(by~(\ref{equation101_u_1}))}\\
&\leq  \sum_{(k,\{ i,j\})\in TA} \min\{3 w(d_k, s_i)+ w(d_k, s_j), 3 w(d_k, s_j)+ w(d_k, s_i) \}  \\
&\leq  \sum_{(k,\{ i,j\})\in M^*} \min\{3 w(d_k, s_i)+ w(d_k, s_j), 3 w(d_k, s_j)+ w(d_k, s_i) \}  \\
&\leq  \sum_{(k,\{ i,j\})\in M^*} 4 cost(k,\{i,j\})  \ \ \text{(by~(\ref{inequation2_sum_s=t}))}\\
&\leq   4 \ cost(M^*)
  \end{split}
\end{equation*}
The first inequality holds by $u_{ij}=u_{ji}=w(s_i,s_j)$ and triangle inequality, and the second inequality holds since TA is a minimum weight perfect matching in algorithm TA$(2)$. 
\end{proof}

 \begin{theorem}
\label{the:allapp_the_2}
$\CA(2, \mu)$ is a 10/7-approximation algorithm for $\CS_{sum,s=t}$. 
\end{theorem}
\begin{proof}
\begin{equation*}
   \begin{split}
 \label{equation_ca2_sum_s=t}
&7 cost(\text{CA($2,\mu$)})\le  6cost(\text{MA($2,\mu$)}) + cost(\text{TA($2$)})  \\
&\leq  \sum_{(k,\{ i,j\})\in M^*} 6(w(s_i,s_j)+(w(d_k, s_i), w(d_k, s_j))/2) \\
& \ \ \ \ +\sum_{(k,\{ i,j\})\in M^*} \min\{3w(d_k, s_i)+ w(d_k, s_j), 3w(d_k, s_j)+ w(d_k, s_i)\}\\
&\leq  \sum_{(k,\{ i,j\})\in M^*} (6w(s_i,s_j)+6\min\{w(d_k, s_i), w(d_k, s_j)\} + 4\max\{w(d_k, s_i), w(d_k, s_j)\}) \\
&\leq  \sum_{(k,\{ i,j\})\in M^*} (10w(s_i,s_j)+10\min\{w(d_k, s_i), w(d_k, s_j)\} )\ \ \text{(by the triangle inequality)}\\
&\leq  \sum_{(k,\{ i,j\})\in M^*} 10cost(k,\{i,j\}) \ \ \text{(by~(\ref{equation101}))}   \\
&\leq   10 \ cost(M^*)
  \end{split}
\end{equation*}
The second inequality holds since MA is a minimum weight perfect matching in algorithm MA$(2, \mu)$ and TA is a minimum weight perfect matching in algorithm TA$(2)$.
\end{proof}
 
\subsection{Remaining proofs of results for  $\CS_{lat}$}\label{A5_lat}
In $\CS_{lat}$, for each $(k,\{i,j\})$, we have:
\begin{eqnarray}\label{inequation1_lat} wait(k,\{i,j\})= \min\{2w(d_k, s_i)+ \mu_{ij}, 2w(d_k, s_j)+ \mu_{ji}\}, \mbox{ and}, \\
\label{inequation2_lat} wait(k,\{i,j\})\ge \max\{ w(d_k, s_i, t_i)+ w(d_k, s_j, t_j), w(d_k, t_i)+w(d_k, t_j)\}.\end{eqnarray}

\begin{lemma}
\label{lem:allapp_lem_7}
$\MA(2, \mu)$ is a $2$-approximation algorithm for $\CS_{lat}$.  Moreover, there exists an instance $I$ for which  wait$(\MA(I))$=2 wait$(M^*(I))$.
\end{lemma}
\begin{proof}
\begin{equation}
   \begin{split}
 \label{equation_ma2_lat}
wait&(\text{MA$(2, \mu)$})=\sum_{(k,\{i,j\})\in \text{MA}} \min\{2w(d_k,  s_i) + \mu_{ij}, 2w(s_k, s_j)+\mu_{ji}\} \ \  \text{(by~(\ref{equation102_u_1}))}  \\
&\leq  \sum_{(k,\{ i,j\})\in M^*} \frac{2w(d_k, s_i)+ \mu_{ij}+ 2w(d_k, s_j)+ \mu_{ji}}{2}\\
&\leq  \sum_{(k,\{ i,j\})\in M^*}   \frac{4\min \{w(d_k, s_i), w(d_k, s_j)\}+ \mu_{ij}+ 2w(s_i, s_j)+ \mu_{ji}}{2} \\
&\leq  \sum_{(k,\{ i,j\})\in M^*}  \min\{4w(d_k, s_i)+2 \mu_{ij}, 4w(d_k, s_j)+2 \mu_{ji} \} \ \ \text{(by~\cref{claim_uu})} \\
&\leq  \sum_{(k,\{ i,j\})\in M^*} 2wait(k,\{i,j\})  \ \ \text{(by~(\ref{inequation1_lat}))} \\
&\leq   2 \ wait(M^*)
  \end{split}
\end{equation}

The first inequality holds since MA is a minimum weight perfect matching in algorithm MA$(2, \mu)$, and by the triangle inequality.
 To see that equality  may hold in (\ref{equation_ma2_lat}), consider the subgraph induced by the nodes $(k_1,s_3)$, $(s_1,s_2)$, and $(k_2,s_4)$ in \cref{fig405}.
\end{proof}

\begin{lemma}
\label{lem:allapp_lem_8}
$\TA(2)$ is a $2$-approximation algorithm for $\CS_{lat}$.  Moreover, there exists an instance $I$ for which  wait$(\TA(I))$=2 wait$(M^*(I))$.
\end{lemma}
\begin{proof}
\begin{equation}
   \begin{split}
 \label{equation_ta2_lat}
&wait(\text{TA($2$)})=\sum_{(k,\{i,j\})\in \text{TA}} \min\{2w(d_k,  s_i) + \mu_{ij}, 2w(s_k, s_j)+\mu_{ji}\}
 \ \  \text{(by~(\ref{equation102_u_1}))}  \\
&\leq  \sum_{(k,\{ i,j\})\in M^*} \min\{2 w(d_k, s_i, t_i)+ w(t_i, d_k, s_j, t_j),2 w(d_k, s_j, t_j)+ w(t_j, d_k, s_i, t_i) \}  \\
&\leq  \sum_{(k,\{ i,j\})\in M^*} \frac{3 w(d_k, s_i, t_i)+3 w(d_k, s_j, t_j)+ w(t_i, d_k)+ w(t_j, d_k)}{2}  \\
&\leq  \sum_{(k,\{ i,j\})\in M^*} 2wait(k,\{i,j\})  \ \ \text{(by~(\ref{inequation2_lat}))}\\
&\leq   2 \ wait(M^*)
  \end{split}
\end{equation}
 
The first inequality holds since TA is a minimum weight perfect matching in algorithm TA$(2)$ and by the triangle inequality. 
To see that equality may hold in (\ref{equation_ta2_lat}), consider the subgraph induced by the nodes $(s_5,s_6)$, $(k_3,k_4)$, and $(s_7,s_8)$ in \cref{fig405}.  \end{proof}

 \begin{lemma}
\label{claim_uuu}
For each request pair $\{i,j\} \in R$, we have
\[u_{ij}+u_{ji}+2w(s_i, s_j)\le \min\{4\mu_{ij}, 4\mu_{ji}\}.\]
\end{lemma}
\begin{proof}
We first prove $u_{ij}+u_{ji}+2w(s_i, s_j)\le 4\mu_{ij}$ by distinguishing three cases.

\textit{Case 1} $ \mu_{ij}=2w(s_i, t_i)+w(t_i,s_j, t_j)$. Observe that $u_{ij}\le w(s_i, t_i, s_j, t_j)$ and $u_{ji}\le w(s_j, t_j, s_i, t_i) $. By the triangle inequality $ w(s_i, t_j)\le w(s_i, t_i, s_j, t_i) $ and $ w(s_i, s_j)\le w(s_i, t_i, s_j) $, we have $u_{ij}+u_{ji}+2w(s_i, s_j) \le 5w(s_i, t_i)+ 4w(t_i, s_j)+2w(s_j, t_j) \le 4 \mu_{ij}$.

\textit{Case 2} $ \mu_{ij}=2w(s_i,s_j,t_i)+w(t_i, t_j)$. Observe that $u_{ij}\le w(s_i, s_j, t_i, t_j)$ and $u_{ji}\le w(s_j, s_i, t_i, t_j) $. By the triangle inequality $ w(s_i, t_i)\le w(s_i, t_j, t_i) $, we have $u_{ij}+u_{ji}+2w(s_i, s_j)  \le 5w(s_i, s_j)+ 2w(s_j, t_i)+2w(t_i, t_j) \le 4 \mu_{ij}$.

\textit{Case 3} $ \mu_{ij}=2w(s_i,s_j, t_j)+w(t_j, t_i)$.  Observe that $u_{ij}\le w(s_i, s_j, t_j, t_i)$ and $u_{ji}\le w(s_j, s_i, t_j, t_i) $. By the triangle inequality $ w(s_i, t_j)\le w(s_i, s_j, t_j) $, we have $u_{ij}+u_{ji}+2w(s_i, s_j) \le 5w(s_i, s_j)+ 2w(s_j, t_j)+2w(t_j, t_i) \le 4 \mu_{ij}$.

The proof of   $u_{ij}+u_{ji}+2w(s_i, s_j)\le 4\mu_{ji}$ follows from symmetry.
 \end{proof}

\begin{lemma}
\label{lem:allapp_lem_9}
$\MA(1,u)$ is a $3$-approximation algorithm for $\CS_{lat}$.
\end{lemma}
\begin{proof}
\begin{equation}
   \begin{split}
 \label{equation_ma1_lat}
&wait(\text{MA($1,u$)})= \sum_{(k,\{i,j\})\in \text{MA}} \min\{2w(d_k,  s_i) + \mu_{ij}, 2w(s_k, s_j)+\mu_{ji}\}  \ \  \text{(by~(\ref{equation102_u_1}))}  \\
&\leq  \sum_{(k,\{ i,j\})\in M^*} \max\{\mu_{ij}, \mu_{ji}\}+\min\{2w(d_k,  s_i), 2w(s_k, s_j)\} \\
&\leq  \sum_{(k,\{ i,j\})\in M^*} (u_{ij}+u_{ji}+ w(d_k, s_i)+ w(d_k, s_j)) \\
&\leq  \sum_{(k,\{ i,j\})\in M^*} (u_{ij}+u_{ji}+ 2\min \{w(d_k, s_i), w(d_k, s_j)\}+ w(s_i, s_j)) \\
&\leq  \sum_{(k,\{ i,j\})\in M^*} (2\min\{w(d_k, s_i), w(d_k, s_j)\}+3\min\{\mu_{ij},\mu_{ji}\}) \ \ \text{(by~\cref{claim_uuu}) }\\
&\leq  \sum_{(k,\{ i,j\})\in M^*}  3 wait(k,\{i,j\}) \ \ \text{(by~(\ref{inequation1_lat_s=t}) and~(\ref{inequation2_lat_s=t})) }  \\
&\leq   3 \ wait(M^*)
  \end{split}
\end{equation}
The second inequality follows from~\cref{lemma:ma_min}, and  the second inequality holds by the triangle inequality.
\end{proof}

\begin{lemma}
\label{lem:allapp_lem_10}
$\TA(1)$ is a $3$-approximation algorithm for $\CS_{lat}$.
\end{lemma}
\begin{proof}
\begin{equation}
   \begin{split}
 \label{equation_ta1_lat}
&wait(\text{TA($1$)})= \sum_{(k,\{i,j\})\in \text{TA}} \min\{2w(d_k,  s_i) + \mu_{ij}, 2w(s_k, s_j)+\mu_{ji}\}
 \ \  \text{(by~(\ref{equation102_u_1}))}  \\
&\leq  \sum_{(k,\{ i,j\})\in M^*} 2 \min\{w(d_k, s_i, t_i, d_k, s_j, t_j), w(d_k, s_j, t_j, d_k, s_i, t_i) \}  \\
&\leq  \sum_{(k,\{ i,j\})\in M^*} 3 wait(k,\{i,j\})  \ \ \text{(by the triangle inequality and case analysis)}\\
&\leq   3 \ wait(M^*)
  \end{split}
\end{equation}
The first inequality holds since TA is a minimum weight perfect matching in algorithm TA$(1)$ and by the triangle inequality.
\end{proof}

 \begin{theorem}
\label{the:allapp_the_3}
$\CA(1, u)$ is a 3-approximation algorithm for $\CS_{lat}$. 
\end{theorem}
According to~(\ref{equation_ma1_lat}) and~(\ref{equation_ta1_lat}), we have wait(CA(I))= $\min\{\text{wait(MA(I))},  \text{wait(TA(I))}\}  \le$ 3 wait($M^*$(I))  for any instance $I$.

\subsection{Remaining proofs of results for  $\CS_{lat, s=t}$}\label{A5_lat_st}
In $\CS_{lat, s=t}$, for each $(k,\{i,j\})$, we have:
\begin{eqnarray}\label{inequation1_lat_s=t} wait(k,\{i,j\})\ge 2\min\{w(d_k, s_i), w(d_k, s_j)\}+ w(s_i, s_j),\mbox{ and }\\
\label{inequation2_lat_s=t} wait(k,\{i,j\})\ge w(d_k, s_i)+ w(d_k, s_j).\end{eqnarray}

\begin{lemma}
\label{lem:allapp_lem_11}
$\MA(2,\mu)$ is a $2$-approximation algorithm for $\CS_{lat,s=t}$.  Moreover, there exists an instance $I$ for which  wait$(\MA(I))$2 wait$(M^*(I))$.
\end{lemma}
\begin{proof}
\begin{equation}
   \begin{split}
 \label{equation_ma2_lat_s=t}
wait&(\text{MA($2,\mu$)})=\sum_{(k,\{i,j\})\in \text{MA}} \min\{2w(d_k,  s_i) + \mu_{ij}, 2w(s_k, s_j)+\mu_{ji}\}  \ \  \text{(by~(\ref{equation102_u_1}))}   \\
&\leq  \sum_{(k,\{ i,j\})\in M^*} (w(s_i,s_j)+ w(d_k, s_i)+ w(d_k, s_j))  \\
&\leq  \sum_{(k,\{ i,j\})\in M^*} 2wait(k,\{i,j\})  \ \ \text{(by~(\ref{inequation1_lat_s=t}) and~(\ref{inequation2_lat_s=t})) } \\
&\leq   2 \ wait(M^*)
  \end{split}
\end{equation}
The first inequality holds since MA is a minimum weight perfect matching in algorithm MA$(2,\mu)$, and by $\mu_{ij}=w(s_i,s_j) $.
 To see that equality  may hold in (\ref{equation_ma2_lat_s=t}), consider the subgraph induced by the nodes $(k_1,s_3)$, $(s_1,s_2)$, and $(k_2,s_4)$ in \cref{fig405}.
\end{proof}

\begin{lemma}
\label{lem:allapp_lem_12}
$\TA(2)$ is a $2$-approximation algorithm for $\CS_{lat,s=t}$.  Moreover, there exists an instance $I$ for which  wait$(\TA(I))$=2 wait$(M^*(I))$.
\end{lemma}
\begin{proof}
We have:
\begin{equation}
   \begin{split}
 \label{equation_ta2_lat_s=t}
wait&(\text{TA($2$)})=\sum_{(k,\{i,j\})\in \text{TA}} \min\{2w(d_k,  s_i) + \mu_{ij}, 2w(s_k, s_j)+\mu_{ji}\}
 \ \  \text{(by~(\ref{equation102_u_1}))}   \\
&\leq  \sum_{(k,\{ i,j\})\in M^*} \min\{3 w(d_k, s_i)+ w(d_k, s_j),3 w(d_k, s_j)+ w(d_k, s_i) \}  \\
&\leq  \sum_{(k,\{ i,j\})\in M^*} 2wait(k,\{i,j\})  \ \ \text{(by~(\ref{inequation2_lat_s=t}))}\\
&\leq   2 \ wait(M^*).
  \end{split}
\end{equation}
The first inequality holds since TA is a minimum weight perfect matching in algorithm TA$(2)$, and by $\mu_{ij}=w(s_i,s_j)$.
 To see that equality  may hold in (\ref{equation_ta2_lat_s=t}), consider the subgraph induced by the nodes $(s_5,s_6)$, $(k_3,k_4)$, and $(s_7,s_8)$ in \cref{fig405}.
\end{proof}

\begin{lemma}
\label{lem:allapp_lem_13}
 $\MA(1,u)$ is a $2$-approximation algorithm for $\CS_{lat,s=t}$.
\end{lemma}
\begin{proof}
\begin{equation*}
   \begin{split}
 \label{equation_ma1_lat_s=t}
wait&(\text{ MA($1,u$)})=\sum_{(k,\{i,j\})\in \text{MA}} \min\{2w(d_k,  s_i) + \mu_{ij}, 2w(d_k, s_j)+\mu_{ji}\}  \ \  \text{(by~(\ref{equation102_u_1}))}   \\
& \leq  \sum_{(k,\{ i,j\})\in  M^*} \min\{2w(d_k,  s_i) + w(s_i,s_j), 2w(d_k, s_j)+w(s_i,s_j)\} \\
&\leq  \sum_{(k,\{ i,j\})\in M^*} (w(s_i,s_j)+ w(d_k, s_i)+ w(d_k, s_j)) \\
&\leq  \sum_{(k,\{ i,j\})\in M^*} 2wait(k,\{i,j\}) \ \ \text{(by~(\ref{inequation1_lat_s=t})) and (\ref{inequation2_lat_s=t})) }  \\
&\leq   2 \ wait(M^*)
  \end{split}
\end{equation*}
The first inequality holds since MA is a minimum weight perfect matching in algorithm MA$(1,u)$, and by $\mu_{ij}=u_{ij}=w(s_i,s_j)$.
\end{proof}

\begin{lemma}
\label{lem:allapp_lem_14}
$\TA(1)$ is a $2$-approximation algorithm for $\CS_{lat,s=t}$.
\end{lemma}
\begin{proof}
\begin{equation*}
   \begin{split}
 \label{equation_ta1_lat_s=t}
&wait(\text{TA($1$)})=\sum_{(k,\{i,j\})\in \text{TA}} \min\{2w(d_k,  s_i) + \mu_{ij}, 2w(s_k, s_j)+\mu_{ji}\}
 \ \  \text{(by~(\ref{equation102_u_1}))}   \\
& \le   4/3 \sum_{(k,\{i,j\})\in \text{TA}}  \min\{2w(d_k,  s_i) +w(d_k,s_j), 2w(s_k, s_j)+w(d_k,s_i)\} \\
&\leq   4/3 \sum_{(k,\{ i,j\})\in M^*} \min\{2 w(d_k, s_i)+ w(d_k, s_j),2 w(d_k, s_j)+ w(d_k, s_i) \}  \\
&\leq  \sum_{(k,\{ i,j\})\in M^*}  \min\{8/3 w(d_k, s_i)+ 4/3 w(d_k, s_j),8/3 w(d_k, s_j)+ 4/3 w(d_k, s_i) \}  \\
&\leq  \sum_{(k,\{ i,j\})\in M^*} 2wait(k,\{i,j\})  \ \ \text{(by~(\ref{inequation2_lat_s=t}))}\\
&\leq   2 \ wait(M^*)
  \end{split}
\end{equation*}
The first inequality holds since $\mu_{ij}=w(s_i,s_j)$ and the triangle inequality, and the second inequality holds since TA is a minimum weight perfect matching in algorithm TA$(1)$.
\end{proof}

 \begin{theorem}
\label{the:allapp_the_4}
$\CA(1,u)$ is a 8/5-approximation algorithm for $\CS_{lat,s=t}$. 
\end{theorem}
\begin{proof}
\begin{equation*}
   \begin{split}
 \label{equation_ca1_lat_s=t}
&5 wait(\text{CA($1,u$)})\le  2 wait(\text{MA($1,u$)})+ 3 wait(\text{TA($1$)})  \\
&\leq  \sum_{(k,\{ i,j\})\in \text{MA}} 2(w(s_i,s_j)+2\min\{w(d_k, s_i), w(d_k, s_j)\})\\
& \ \ \ \  +4 \sum_{(k,\{ i,j\})\in \text{TA}} \min\{2w(d_k, s_i)+ w(d_k, s_j), 2w(d_k, s_j)+ w(d_k, s_i)\}\\
&\leq  \sum_{(k,\{ i,j\})\in M^*} 2(w(s_i,s_j)+2\min\{w(d_k, s_i), w(d_k, s_j)\})\\
& \ \ \ \  +4 \sum_{(k,\{ i,j\})\in M^*} \min\{2w(d_k, s_i)+ w(d_k, s_j), 2w(d_k, s_j)+ w(d_k, s_i)\}\\
&\leq  \sum_{(k,\{ i,j\})\in M^*} (2w(s_i,s_j)+ 2w(d_k, s_i)+ 2w(d_k, s_j) + 8w(d_k, s_i)+  4 w(d_k, s_j)) \\
&\leq  \sum_{(k,\{ i,j\})\in M^*} 8wait(k,\{i,j\}) \ \ \text{(by~(\ref{inequation1_lat_s=t}) and (\ref{inequation2_lat_s=t}))}   \\
&\leq   8 \ wait(M^*)
  \end{split}
\end{equation*}
The first inequality holds since $\mu_{ij}=w(s_i,s_j)$ and the triangle inequality, and the second inequality holds since MA is a minimum weight perfect matching in algorithm MA$(1,u)$, and since TA is a minimum weight perfect matching in algorithm TA$(1)$.
\end{proof}

\end{document}